\documentclass{article}
\usepackage[a4paper, total={6.5in, 10in}]{geometry}
\usepackage{tikz}
\usepackage{xypic}
\usetikzlibrary{shapes}
\usetikzlibrary{arrows.meta}
    \usepackage{graphicx}
     \usepackage{wrapfig}
\usepackage{amsmath,amssymb,amsthm}

\usepackage{mathrsfs}
\usepackage{bbm}
\usepackage{comment}
\usepackage{url}

\newcommand{\csp}{\textrm{CSP}}
\newcommand{\Pol}{\textrm{Pol}}
\newcommand{\Aut}{\textrm{Aut}}
\newcommand{\Age}{\textrm{Age}}

\newcommand{\structA}{\mathbb{A}}
\newcommand{\structB}{\mathbb{B}}
\newcommand{\structC}{\mathbb{C}}
\newcommand{\structD}{\mathbb{D}}

\newcommand{\subsetA}{\mathcal{A}}
\newcommand{\subsetB}{\mathcal{B}}
\newcommand{\subsetC}{\mathcal{C}}
\newcommand{\subsetD}{\mathcal{D}}

\newcommand{\orbitA}{\mathbf{A}}
\newcommand{\orbitB}{\mathbf{B}}
\newcommand{\orbitC}{\mathbf{C}}
\newcommand{\orbitD}{\mathbf{D}}
\newcommand{\orbitE}{\mathbf{E}}

\newcommand{\orbitM}{\mathbf{M}}
\newcommand{\orbitN}{\mathbf{N}}
\newcommand{\orbitO}{\mathbf{O}}
\newcommand{\orbitP}{\mathbf{P}}

\newcommand{\constraint}{\mathbbm{C}}

\newcommand{\forbA}{\mathcal{F}_{\structA}}
\newcommand{\maxboundA}{\mathbb{L}_{\structA}}

\newcommand{\instance}{\mathcal{I}}
\newcommand{\instanceJ}{\mathcal{J}}
\newcommand{\Var}{\mathcal{V}}

\newcommand{\bipartite}{\mathfrak{B}}
\newcommand{\vertices}{\mathbf{V}}

\newcommand{\graphinstance}{\mathcal{G}_{\mathcal{I}}}
\newcommand{\maximal}{\mathcal{M}}

\newcommand{\graphinstancemaximal}{\mathcal{G}_{\instance}^{\maximal}}
\newcommand{\componentmin}{\mathcal{N}_{\textrm{min}}}

\newcommand\restr[2]{{
  \left.\kern-\nulldelimiterspace 
  #1 
  \vphantom{\big|} 
  \right|_{#2} 
  }}

\newtheorem{theorem}{Theorem}
\newtheorem{conjecture}{Conjecture}
\newtheorem{definition}{Definition}
\newtheorem{lemma}{Lemma}
\newtheorem{proposition}{Proposition}
\newtheorem{observation}{Observation}
\newtheorem{claim}{Claim}
\newtheorem{corollary}{Corollary}
\newtheorem{example}{Example}

\begin{document}

\title{ Quasi Directed   J\'{o}nsson Operations Imply Bounded Width 
\\ 
(For fo-expansions of symmetric  binary cores with free amalgamation)}


\author{Micha\l\ Wrona\footnote{This author is partially supported by National Science Centre, Poland grant number 2020/37/B/ST6/01179.}\\
Jagiellonian University\\ 
Krak\'{o}w\\
Poland\\
\url{michal.wrona@uj.edu.pl} 
}







\maketitle
\begin{abstract}
	Every $\csp(\structB)$ for a finite structure $\structB$ is either in P or it is NP-complete but the proofs of the finite-domain CSP dichotomy by Andrei Bulatov and Dimitryi Zhuk not only show the computational complexity separation but also confirm the algebraic tractability conjecture stating that tractability origins from a certain system of operations preserving B. The establishment of the dichotomy was in fact preceded by a number of similar results for stronger conditions of this type, i.e. for system of operations covering not necessarily all tractable finite-domain CSPs.

A similar, infinite-domain algebraic tractability conjecture is known for first-order reducts of countably infinite finitely bounded homogeneous structures and is currently wide open. In particular, with an exception of a quasi near-unanimity operation there are no known systems of operations implying tractability in this regime. This paper changes the state-of-the-art and provides a proof that a chain of quasi directed J\'{o}nsson operations imply tractability and bounded width for a large and natural class of infinite structures.  
\end{abstract}

\section{Introduction}

\label{sect:intro}
Constraint Satisfaction Problems form a large class of computational problems
whose complexity has been studied separately as well as within the formalism $\csp(\structB)$ where $\structB$ is a relational structure over domain $A$.
An instance $\instance$ of $\csp(\structB)$  consists of a set of constraints $\constraint := ((x_1, \ldots, x_k),R)$ 
formed out of a tuple of variables and a $k$-ary relation $R$ in $\structB$. The question is whether there is a solution to $\instance$, i.e., an assignment of elements from $A$ to the variables such that for all
$\constraint$ we have $f((x_1, \ldots, x_k)) \in R$. 

The formalism $\csp(\structB)$ for a finite $\structB$ captures a number of natural problems. Indeed, if $\structB$ is a graph $\mathcal{H}$, then we cope with an $\mathcal{H}$-coloring problem~\cite{HellNesetril}, in  particular with a $k$-colouring problem. If the domain of $\structB$ 
has only two elements, then $\csp(\structB)$ is a Boolean satisfiability problem such as $k$-SAT or NAE-SAT~\cite{Schaefer78}. There are also structures $\structB$ whose $\csp(\structB)$ comes down to solving a system of equations over a finite field~\cite{IdziakMMVW10}.
The variety of natural NP-complete and polynomially solvable problems 
within this formalism raises a number of natural questions including:
\begin{enumerate}
\item Is there a single simple technique of showing NP-completness for $\csp(\structB)$?
\item Is there a single algorithm that solves all polynomially tractable $\csp(\structB)$?
\item Are all $\csp(\structB)$ either NP-complete or polynomially tractable?  
\end{enumerate}

The conjecture that the answer to the last question is affirmative was known as the Feder-Vardi Conjecture~\cite{FederV98}. However, there was no reasonable progress until the questions were reformulated in the language of universal algebra~\cite{BulatovJK05}. The immediate consequence of that step was the conditional answer to the first question and a clear algebraic conjecture concerning the second question known as the \emph{algebraic tractability conjecture} saying that a $\csp(\structB)$ is in PTIME if and only if the algebra corresponding to $\structB$ lies in a Taylor variety. Actually, this conjecture has a number of equivalent formulations and the one of which we care the most in this paper says that $\csp(\structB)$ is in PTIME iff $\structB$ is preserved by a Siggers operation~\cite{Siggers}. 
The algebraic tractability conjecture has been confirmed independently by
Bulatov~\cite{Bulatov17} and Zhuk~\cite{Zhuk20}, and hence all the three questions for finite structures $\structB$
have been answered.

Although the Bulatov-Zhuk theorem is a great achievement,
there are still some natural computational problems that can be expressed as $\csp(\structB)$ but only when $\structB$ is infinite. A perfect example is DIGRAPH-ACYCLICITY equivalent to $\csp(\mathbb{Q}; <)$ where $(\mathbb{Q}; <)$ 
is the natural linear order over the rational numbers. That structure as well as  all 
structures with a first-order definition in (first-order reducts of) $(\mathbb{Q}; <)$ a.k.a. \emph{temporal languages} give rise to constraint satisfaction problems of interest in a field of artificial intelligence known as spatial and temporal reasoning, see~\cite{bodirsky_2021} for more examples of natural computational problems expressible as infinite-domain CSPs.  Furthermore, the structure  $(\mathbb{Q}; <)$ is a natural representative of a class of (countably infinite) finitely bounded homogeneous structures.

A $\tau$-structure $\structA$ is \emph{homogeneous} if every local isomorphism between finite substructures of $\structA$ may be extended to an automorphism of $\structA$. It is \emph{finitely bounded} if there exists a finite set of finite $\tau$-structures
$\forbA$ such that a finite structure $\structD$ embeds into $\structA$ if and only if there is no $\structC \in \forbA$
that embeds into $\structD$.
The following conjecture resembling the finite algebraic tractability conjecture has been formulated for cores of finite-signature first-order reducts of finitely-bounded homogeneous structures~\cite{BartoP16}. 
(All definitions that are omitted in the introductionn are provided carefully in preliminaries.)

\begin{conjecture}(\textbf{Infinite-domain Algebraic Tractability Conjecture}) 
\label{conj:inftract}
Let   $\structB$ be the core of a  finite-signature first-order reduct of a finitely bounded homogeneous structure.
If $\structB$ is preserved by a pseudo-Siggers operation, i.e., a 6-ary operation $s$ and some unary operations $\alpha, \beta$
such that
$$\alpha s(x, y, x, z, y, z) \approx \beta s(y, x, z, x, z, y)$$
for all $x,y,z$ in domain,  then $\csp(\structA)$ is solvable in polynomial time.
\end{conjecture}

The above conjecture has been confirmed in a number of special cases, in particular for  temporal languages~\cite{BodirskyK10} or
first-order reducts of homogeneous graphs~\cite{BodirskyP15,BodirskyMPP19}.
It is also known that if there is no pseudo-Siggers operation preserving $\structB$, then $\csp(\structB)$ is NP-complete~\cite{bodirsky_pinsker_pongrácz_2021}.
A simple strategy to attack the above conjecture in full generality would be to provide a polynomial-time algorithm for an infinite-domain $\csp(\structB)$  
which works only under the assumption that $\structB$ is preserved by a pseudo-Siggers operation. 
The reality is, however, that we do not have such algorithms even for operations satisfying much stronger algebraic conditions of that type, called \emph{pseudo minor conditions} or \emph{quasi minor conditions}. The situation is very different for finite-domain
CSPs where a number of algebraic conditions implying tractability has been identified prior to the dichotomy proofs, see~\cite{IdziakMMVW10} and~\cite{BartoK14} for the two most important results of this kind. 
The notable exception in an infinite-domain CSP regime is a quasi near-unanimity operation $f$ satisfying:
$$f(y, x, \ldots, x) \approx \cdots \approx f(x, \ldots,x, y) \approx f(x,\ldots, x)$$
for all $x,y$ in domain, but even in this case the proof~\cite{BodirskyD13}
is a straightforward adaptation of the one for finite structures~\cite{FederV98}.

 The next natural algebraic condition to consider for infinite structures are  chains of quasi directed J\'{o}nsson
operations.

\begin{definition}
\label{def:quasi_directed_Jonsson}
A sequence $(D_1, \ldots , D_n)$ of ternary operations on a set $A$ is
called a \emph{chain of quasi directed J\'{o}nsson operations} if for all $x, y, z \in A$ all of the following hold:

\begin{eqnarray}
\label{eq:D1} D_1(x,x,y) &=& D_1(x,x,x), \\
\label{eq:Di} D_i(x,y,x) &=& D_i(x,x,x) \hspace{40 pt }\textrm{ for all } i \in [n],\\
\label{eq:Dii+1} D_i(x,y,y) &=& D_{i+1}(x,x,y) \hspace{30 pt} \textrm{ for all } i \in [n-1],\\
\label{eq:Dn} D_n(x,y,y) &=& D_n(y,y,y).
\end{eqnarray} 
\end{definition}

Every relational structure preserved by a quasi near-unanimity operation is also preserved by a chain of quasi directed J\'{o}nsson operations, see~\cite{bodirsky_2021} for a simple proof.
Another, even weaker system of operations worth considering is  a chain of \emph{quasi J\'{o}nsson operations}. 

A finite-signature finite structure is preserved by a 
\begin{itemize}
    \item near-unanimity operation iff 
    \item by a chain of directed J\'{o}nsson operations~\cite{Kazda2018} iff
    \item  by a chain of J\'{o}nsson operations\cite{barto_2013}. 
\end{itemize} 
One should however keep in mind that a direct proof  of tractability for J\'{o}nsson operations~\cite{BartoK09} came prior to the equivalence of the three-conditions.

The main result of this paper, explained in details in the following subsection, states that every first-order expansion $\structB$ of a finitely-bounded homogeneous symmetric binary core (all relations are binary and symmetric)
whose age has free amalgamation and which is preserved by a chain of directed quasi  J\'{o}nsson operations has bounded width, and in consequence $\csp(\structB)$ may be solved in polynomial time by establishing local consistency. 

Examples of symmetric binary structures are the  countably infinite homogeneous universal graph
a.k.a. the random graph $\mathbb{G}$ or the countably infinite 
homogeneous universal graph omitting cliques of size $k$, for any $k \geq 3$, known as the $k$-th Henson graph $\mathbb{H}_k$. The complexity classifications of CSPs for first-order reducts of these structures have been obtained  in~\cite{BodirskyP15} and respectively in~\cite{BodirskyMPP19}.
The authors of the latter paper actually suggest that a reasonable intermediate step towards confirming Conjecture~\ref{conj:inftract}  might be to prove it for first-order reducts of structures whose  age has \emph{free amalgamation}. Here we make a major step towards completing that goal, which is much more challenging, even for symmetric binary cores, than the classification for homogeneous graphs presented in that paper. Homogeneous graphs have been classified in~\cite{LachlanWoodrow} --- there are only few types of these. There are no such results for symmetric binary cores whose age has free amalgamation except for some partial results over multi-graphs with three kinds of edges~\cite{cherlin_2022_2}.

\subsection{The Main Result}
\label{subsect:mainresult}
In this section we present the main result of this paper. 

\subsubsection{Fra\"{i}ss\'{e} Amalgamation and Free Amalgamation}

Let $\structB_1, \structB_2$ be two structures over 
the same signature $\tau$ such that all the common elements are the elements of $\structA$. We say that $\structC$ is an \emph{amalgam} of $\structB_1$ and $\structB_2$ over $\structA$
if for
$i = 1, 2$ there are embeddings $f_i: \structB_i
\rightarrow \structC$  such that $f_1(a) = f_2(a)$ for every $a \in \structA$.
An isomorphism-closed class $\mathscr{C}$ of relational $\tau$-structures
 has the \emph{amalgamation property} if 
$\mathscr{C}$ is closed under taking amalgams. 
A class of finite $\tau$-structures that
contains at most countably many non-isomorphic structures, has the amalgamation
property, and is closed under taking induced substructures and isomorphisms is called
an \emph{amalgamation class}.
It is known that there is a deep connection between amalgamation classes and homogeneous structures. The age of a homogeneous $\tau$-structure $\structA$ is the class of all finite $\tau$-structures that embed into $\structA$. First of all, it is known that the age of a homogeneous structure is an amalgamation class. But what is more surprising is that out of any amalgamation class one can construct a unique, up to isomorphism, homogeneous structure.

\begin{theorem} (Fra\"{i}ss\'{e})
Let $\tau$ be a countable relational signature
and let $\mathscr{C}$ be an amalgamation class of $\tau$-structures. Then there is a homogeneous and
at most countable $\tau$-structure $\structA$ whose age equals $\mathscr{C}$. The structure $\structA$ is unique up to
isomorphism, and called the Fra\"{i}ss\'{e}-limit of $\mathscr{C}$.
\end{theorem}

For example, we have that $(\mathbb{Q}; <)$ is the Fra\"{i}ss\'{e}-limit
of the class of all finite linear orders, $\mathbb{G}$ is the 
Fra\"{i}ss\'{e}-limit
of the class of all finite graphs and that $\mathbb{H}_k$ is the 
Fra\"{i}ss\'{e}-limit
of the class of all finite graphs omitting $k$-cliques.
In fact, in the two last cases as an amalgam of the graphs $G_1$ and $G_2$ one can take the union  $G_1 \cup G_2$ of $G_1$ and $G_2$. We amalgamate without imposing new edges. Observe that it is in general not the case for linear orders when one usually have to add additional arcs to impose transitivity.

In general, we say that $\structB_1 \cup \structB_2$ is a \emph{free amalgam} of $\structB_1$ and $\structB_2$ over $\structB_1 \cap \structB_2$. A class of structures $\mathscr{C}$ for which taking free amalgams suffices is the class with the free amalgamation property. For the sake of simplicity we will simply say say that the class has free amalgamation. 

 A more general recipe of obtaining a finitely bounded homogeneous structure over a signature $\tau= \{ \orbitA, \orbitB \}$  whose age has  free-amalgamation  is given in~\cite{cherlin_2022_2}. First take any finite set of finite $2$-graphs $\forbA$ with 
 $\orbitA$-edges and $\orbitB$-edges: every two different vertices are connected by an $\orbitA$-edge or a $\orbitB$-edge. The class of finite structures $\mathscr{C}_{\structA}$ over $\tau$ omitting the $2$-graphs in $\forbA$ has free amalgamation and the Fra\"{i}ss\'{e} limit $\structA$ of $\mathscr{C}_{\structA}$ is the desired $\tau$-structures. 

Although, it is not the case that the age of all homogeneous structures has free amalgamation, it is a natural restriction for structures considered in Conjecture~\ref{conj:inftract}.




\subsection{Formulation of the Main Result}
\label{sect:mainresultform}

We will prove the tractability of $\csp(\structB)$ for templates under consideration by employing an algorithm establishing $(k,l)$-minimality for $(k \leq l)$. This algorithm is a slight variation of the better known procedure for enforcing $(k,l)$-consistency but better suited for measuring the level of consistency needed to solve a problem~\cite{Barto16}.
Roughly speaking, it propagates the local information about $k$ variables in context of $l$ variables through the structure of a CSP instance until a fixed-point is reached.
Sometimes it only narrows the search space, but in the case of $2$-coloring, $2$-SAT, Horn-SAT or the already mentioned DIGRAPH ACYCLICITY, the procedure simply decides whether a given instance of the problem has a solution. 

We say that a structure $\structB$ has relational width $(k,l)$ if  establishing $(k,l)$-minimality on an instance of  $\csp(\structB)$ decides if it has a solution. A structure $\structB$ has bounded (relational) width
if it has relational width $(k,l)$ for some $(k \leq l)$. 

\smallskip
\noindent
Here comes the main result of this paper.

\begin{theorem}
\label{thm:main}
Let $\structA$ be a finitely bounded homogeneous symmetric binary core whose age has free amalgamation and  $\structB$  a first-order expansion of $\structA$ preserved by a chain of  quasi directed J\'{o}nsson operations. Then $\csp(\structB)$ has  relational width $(2, \maxboundA)$ where $\maxboundA$ is the size of the largest forbidden structure in $\forbA$ but  not smaller than $3$. In particular, $\csp(\structB)$ is in PTIME.
\end{theorem}

\subsection{Bounded (Strict) Width Collapses}
\label{sect:collapses}

A structure $\structB$ has strict width $k$ if every non-trivial $(k,k+1)$-minimal instance of $\csp(\structB)$ not only has a solution but also  every its partial solution may be extended to a total solution. It is known that an $\omega$-categorical structure (all structures in this paper are $\omega$-categorical, for a definition see~\cite{bodirsky_2021}) has 
strict width $k$ if and only if it is preserved by a $(k+1)$-ary quasi near-unanimity operation.
Thus, any first-order expansion $\structB$ of a finitely bounded homogeneous symmetric binary core  whose age has free amalgamation and which is preserved by a $(k+1)$-ary quasi near-unanimity polymorphism has relational width $(k, k+1)$. Theorem~\ref{thm:main} implies a stronger result.

\begin{corollary}
\label{cor:collapse}
Let $\structA$ be a finitely bounded homogeneous symmetric binary core  whose age has free amalgamation and  $\structB$  a first-order expansion of $\structA$ with bounded strict width. Then $\structB$ has  relational width $(2, \maxboundA)$ where $\maxboundA$ is the size of the largest forbidden structure in $\forbA$ but not smaller than $3$. 
\end{corollary}

\begin{proof}
The corollary follows by Theorem~\ref{thm:main} and Proposition 6.9.10 in~\cite{bodirsky_2021}.
\end{proof}


A similar result has been shown for first-order reducts of homogeneous graphs in~\cite{Wrona20STACS} and first-order expansions of (not necessarily symmetric) binary cores which do not forbid any substructures of size 3, 4, 5, or 6 in~\cite{Wrona20LICS}. That kind of results imply that in the considered cases the level of consistency needed to solve $\csp(\structB)$ 
for a first-order reduct (expansion) of $\structA$ depends rather on $\structA$ than on $\structB$. These results are related to the original bounded width hierarchy collapse in~\cite{Barto16} which says that a finite structure with bounded  width has relational width $(1,1)$ or $(2,3)$. The already discussed bounded width collapses for infinite-domain CSP consider only strict width. In~\cite{MottetNPW21}, one can find a result similar to Corollary~\ref{cor:collapse} considering structures with bounded width but not necessarily bounded strict width.

\subsection{Black box approach and canonical polymorphisms}
\label{sect:canonpolim}

An approach of converging to an algorithm for a pseudo-Siggers operation  in Conjecture~\ref{conj:inftract} by finding polynomial-time algorithms for weaker and weaker algebraic conditions is not the only approach that can possibly lead to resolving the conjecture. Indeed, all the algorithmic results in the classification for the first-order reducts of $\mathbb{G}$ or 
$\mathbb{H}_k$ with $k \geq 3$ may be obtained by use of a black box reduction to finite CSP from~\cite{BodirskyM18}. See~\cite{MottetP22} for a neat proof of the latter. That black box reduction is obtained by use of canonical polymorphisms~\cite{BodPin-CanonicalFunctions} and in particular of a canonical pseudo-Siggers operation. The problem with this approach is that already getting from homogeneous graphs to homogeneous hypergraphs makes the use of the simple black box reduction impossible~\cite{Pinsker22,MottetPN23}.  

\subsection{Organisation of the Present Article}
\label{sect:org}

All the notions and definitions that have not been properly introduced in Section~\ref{sect:intro} are defined in Section~\ref{sect:prelim}. In Section~\ref{sect:ImpUniform} we 
treat $\structB$ whose relational clones are called \emph{implicationally uniform}. We show that such $\structB$  have relational width $(2, \maxboundA)$.
 Section~\ref{sect:ImpNonUniform} is dedicated to $\structB$ whose relational clones are implicationally non-uniform --- in that case $\structB$ is not preserved by any chains of quasi directed  J\'{o}nsson operations.
Armed with these results, we give a proof of Theorem~\ref{thm:main} in Section~\ref{sect:mainresult}.

\section{Preliminaries}
\label{sect:prelim}
We use $[n]$ for $\{ 1, \ldots, n\}$ and we write $t[i]$ as a reference to the $i$-th element of an $n$-ary tuple with $i \in [n]$. A structure is usually denote by $\structA, \structB, \structC$ etc. and their corresponding domains by $A,B,C$ etc. 

\subsection{Structures, Relations and Formulas}
\label{sect:structrelform}
All structures $\structA$ considered in this paper are \emph{$\omega$-categorical}, i.e., all countable models of the theory of $\structA$ are isomorphic. It is very well known that an automorphism group of $\structA$, denoted here by $\Aut(\structA)$, of an $\omega$-categorical structure $\structA$ is \emph{oligomorphic}, i.e., there are at most finitely many orbits of $k$-tuples w.r.t. $\Aut(\structA)$ for all $k \in \mathbb{N}$. An \emph{orbit of a $n$-tuple} $t$ is the set 
$\{ s \in A^n \mid \exists \alpha \in \Aut(\structA)~(s[1], \ldots, s[n]) = (\alpha(t[1]), \ldots, \alpha(t[n])) \}$.
Since all orbits in this paper are w.r.t. to the automorphism group $\Aut(\structA)$ of a structure  $\structA$ we will simply say orbits instead of more formally: orbits w.r.t. $\Aut(\structA)$.
  An orbit of a pair is called an \emph{orbital}.  Observe that every orbital except for $=$ is anti-reflexive.

An $\omega$-categorical structure $\structA$ is a \emph{core}, also referred to as a \emph{model-complete core}, if all
of its endomorphisms are elementary self-embeddings, i.e., preserve all first-order formulas definable in $\structA$. 
In this paper, we deal with finitely bounded homogeneous core structures $\structA$ over finite signatures. They are all $\omega$-categorical. We additionally assume that all relations in $\structA$ are both binary and symmetric. Thus, we mainly cope with \emph{finitely bounded homogeneous binary cores} $\structA$ and their  \emph{first-order expansions} $\structB$ that are structures containing all the relations in $\structA$ as well as some other relations with first-order definitions in $\structA$.

For the sake of clarity we will use the same symbol $R$ for a relational symbol in the signature $\tau$ of $\structB$ as well as for the relation $R^{\structA}$. 

Since the age of $\structA$ has free amalgamation, it follows that 
\begin{align}
    \orbitN = \left\{ (a_1, a_2) \in A^2\mid \bigwedge_{R \in \tau}  \neg R^{\structA}(a_1, a_2)  \right\}\nonumber
\end{align}
is nonempty. It is also straightforward to see that $\orbitN$ is an orbital w.r.t. $\Aut(\structA)$.

We intend to use Theorem~4.5.1 in~\cite{bodirsky_2021} and therefore assume that all orbitals are in $\structA$. Hence we need a new definition of free amalgamation known also as free amalgamation over $\orbitN$, see e.g.~\cite{cherlin_2022_2}.

\begin{definition}
\label{def:freemulti}
Let $\structA$ be a finitely bounded homogeneous symmetric binary core over a signature containing $\orbitN$. We say that the age of $\structA$ has \emph{free amalgamation (w.r.t. $\orbitN$)} if for all finite substructures $\structD, \structB_1, \structB_2$ such that $\structD$ is a common substructure of $\structB_1$ and $\structB_2$
of $\structA$ there exists an amalgam $\structC$ of $\structB_1$ and $\structB_2$ over $\structD$ satisfying $(a,b) \in \orbitN^{\structC}$ for all $a \in B_1 \setminus B_2$
and $b \in B_2 \setminus B_1$.
We say that $\structC$ is a \emph{free amalgam (w.r.t. $\orbitN$)} of $\structB_1$ and $\structB_2$ over $\structD$. 
\end{definition}

It is easy to see that the age of a finitely bounded symmetric binary core has free amalgamation according to the definition in Section~\ref{sect:intro} if and only if its expansion with all orbitals has free amalgamation 
w.r.t $\orbitN$. 

\begin{example}
The graph $\mathbb{H}_3 = (H_3, \orbitE)$ pp-defines 
$\orbitN(x_1, x_2) \equiv \exists y~\orbitE(x_1, y) \wedge \orbitE(y, x_2)$.
We have that the age of a $\{ \orbitE \}$-structure $\mathbb{H}_3$ has 
free-amalgamation while the expansion of $\mathbb{H}_3$ seen as a $\{ \orbitE, \orbitN \}$-structure has free amalgamation w.r.t. $\orbitN$. 
\end{example}

For the sake of simplicity we will write that the age of 
a finitely bounded homogeneous symmetric binary core expanded with all orbitals has free amalgamation instead of more properly writing that it has free amalgamation w.r.t. $\orbitN$.

Observe now that every orbit of $k$-tuples w.r.t. $\Aut(\structA)$ where $\structA$ is a finitely bounded homogeneous symmetric binary core may be defined by a conjunction of atomic formulae $\orbitB(x,y)$ where $\orbitB$ is an orbital w.r.t. $\Aut(\structA)$. It explains the following definition.

\begin{definition}
\label{def:BBBtuple}
 Let $\structA$ be a finitely bounded homogeneous symmetric binary core over signature $\tau$. We say that a $k$-ary tuple $t = (t[1], \ldots, t[k])$ 
is a
$$(\orbitB_{1,2},  \ldots \orbitB_{1,k}, \orbitB_{2,3}, \ldots, \orbitB_{2, k}, \ldots, \orbitB_{k-1,k})-\text{tuple},$$
where the indices above are pairs 
$(i < j)$ in $[k]^2$ ordered lexicographically, if $\orbitB_{i,j}$
is an orbital w.r.t. $\Aut(\structA)$ and $(t[i], t[j]) \in \orbitB_{i,j}$ for all $i < j$.
\end{definition}

If we do not care about other orbitals, we often simply write that $t$ is a $(\orbitB_{1,2},  \ldots, \orbitB_{k-1,k})$-tuple requiring only that $(t[1], t[2]) \in \orbitB_{1,2}$ and $(t[k-1], t[k]) \in \orbitB_{k-1,k}$ and we do not care about other $(t[i], t[j])$. 

\begin{example}
A tuple $t \in H_3^3$ where $\mathbb{H}_3 = (H_3, \orbitE, \orbitN)$ is a $(\orbitE, \orbitE, \orbitN)$-tuple if $(t[1], t[2]) \in \orbitE$, $(t[1], t[3]) \in \orbitE$ and $(t[2], t[3]) \in \orbitN$.  The tuple $t$ is in particular a $(\orbitE, \ldots, \orbitN)$-tuple.  
\end{example}

A projection of a formula $\varphi(x_1, \ldots, x_k)$ over the set of free variables $\{ x_1, \ldots, x_k\}$ to variables $\{x_{i_1},\ldots, x_{i_l} \}$ is the  relation
$$R'(y_1, \ldots, y_l) \equiv \exists x_1 \cdots \exists x_{k}~R(x_1, \ldots, x_k) \wedge \bigwedge_{j=1}^l y_j = x_{i_j}.$$
A projection $\Pi_{i_1, \ldots, i_l} R$ of a $k$-ary relation $R$ to coordinates $\{i_1,\ldots, i_l\}$ is a projection of $R(x_1, \ldots, x_k)$ to 
$\{x_{i_1},\ldots, x_{i_l} \}$.
 We extend this notation to   $\{i_1,\ldots, i_l \} \subseteq \{ -k, \ldots, -1, 1, \ldots, k \}$  where the coordinate $-l$ is a shorthand for $k+1-l$.

\begin{definition}
\label{def:plus}
Let  $R$ be an $n$-ary relation with $n \geq 3$ and $\subsetA \subseteq \Pi_{1,2}(R)$.
Then  
$$\subsetA + R =\\ \{ (c,d) \mid \exists t \in R~ (t[1], t[2]) \in \subsetA \wedge (t[-2], t[-1])  = (c,d) \}.$$

\end{definition}

\noindent
We use the above definition mainly in the following context.

\begin{definition}
\label{def:implication}
We say that  $R$  is a $(\subsetA \rightarrow \subsetB)$--implication, which we denote also by $R:\subsetA \rightarrow \subsetB$ if:
\begin{itemize}
    \item $\subsetA \subsetneq \Pi_{1,2}(R)$, $\subsetB \subsetneq \Pi_{-2,-1}(R)$, and
    \item $\subsetB = \subsetA + R$.
\end{itemize}
\end{definition}


We will usually consider pairs of implications satisfying the following condition.

\begin{definition}
\label{def:complementary}
We say that two relations $R_1$ and $R_2$ agree on projections if  
\begin{itemize}
\item $\Pi_{1,2}(R_1) = \Pi_{-2,-1}(R_2)$ and $\Pi_{1,2}(R_2) = \Pi_{-2,-1}(R_1)$.
\end{itemize}

Two implications $R_1:\subsetA \rightarrow \subsetB$ and $R_2: \subsetB \rightarrow \subsetA$ are \emph{complementary} if they agree on projections.
\end{definition}

\subsection{Primitive-Positive Definability and Polymorphisms}
\label{sect:polym}

A \emph{primitive-positive (pp)-formula} is a first-order formula built up exclusively from atomic formulae, conjunction, equality and existential quantifiers.  
A relation $R$ is pp-definable in a structure $\structB$ if it can be defined by a pp-formula. 

\begin{definition}
\label{def:relclone}
A \emph{relational clone} of $\structB$ is the set of all relations pp-definable in $\structB$.
\end{definition}

A \emph{polymorphism} $f:A^m \rightarrow A$ of a structure $\structB$ over domain $A$ is a homomorphism from a power of the structure $\structB$ to $\structB$. A polymorphism $f$ of a relation $R \subseteq A^n$ is a polymorphism of a structure $(A, R)$. In this case we say that $f$ preserves  $\structB$ or $R$.

 There is a deep connection between the polymorphisms of a structure $\structB$, denoted by $\Pol(\structB)$, and the relations pp-definable in that structure.

\begin{theorem}
(\cite{BodirskyN06}) 
\label{thm:Galoisconn}
Let $\structB$ be a countable $\omega$-categorical structure. Then $R$ is preserved by the polymorphisms of $\structB$ if and only if it has a  primitive-positive definition in $\mathbb{A}$.
\end{theorem}

\subsection{Constraints, CSP and Minimality}
\label{sect:CSP}
A constraint $\constraint$ over domain $A$ is a pair $((x_1, \ldots, x_k),R)$ where $(x_1, \ldots, x_k)$ is a tuple of pairwise different variables and $R \subseteq A^k$. 
A set of variables in $\constraint$ called also the scope of $\constraint$ is denoted by $\Var(\constraint)$. A projection $\Pi_{x_{i_1}, \ldots, x_{i_l}} \constraint$ of a constraint $\constraint$ to variables $\{x_{i_1}, \ldots, x_{i_l}\} \subseteq \Var(\constraint)$
is a constraint 
$((x_{i_1}, \ldots, x_{i_l}),\Pi_{i_1, \ldots, i_l} R)$.

\begin{definition}
\label{def:csp}
An instance $\instance$ of $\csp(\structB)$ for a relational structure $\structB$ over domain $A$ and a set of variables $\Var = \{ v_1, \ldots, v_n \}$ is a set of
constraints of the form $((x_1, \ldots, x_k),R)$ such that $
\{ x_1, \ldots, x_k \} \subseteq \Var$ and $R$ is in $\structB$. The question is whether  there is a solution to $\instance$, i.e., a mapping $f: \Var \rightarrow A$ such that for every constraint $((x_1,\ldots, x_k),R)$ it holds that $(f(x_1), \ldots, f(x_k)) \in R$.
\end{definition}

We say that a constraint $((x_1, \ldots, x_k), R)$ is non-trivial if $R$ is different from $\emptyset$ and trivial otherwise. We say that  an instance $\instance$ is non-trivial if all constraints in it are non-trivial, and trivial otherwise.  

\begin{definition}
\label{def:klmininstances}
An instance $\instance$ over variables $\Var = \{ v_1, \ldots, v_n \}$ is $(k,l)$-minimal with $k \leq l$ if 
\begin{enumerate}
\item every $l$-element subset of $\Var$ is contained in the scope of some constraint, and
\item for all $k$-element subsets of variables $\{x_1, \ldots, x_k \}$ and all constraints $\constraint_1, \constraint_2$ whose scope contains $(x_1, \ldots, x_k)$ 
we have that $\Pi_{x_1, \ldots, x_k} \constraint_1$equals $\Pi_{x_1, \ldots, x_k}  \constraint_2$. 
\end{enumerate}
\end{definition}

An algorithm that transforms any instance into a $(k,l)$-minimal instance is straightforward and works in time $O(\left| \Var \right|^m)$ where $m$ is the maximum of $l$ and the largest arity in the signature of $\structB$. Indeed, we simply introduce a new  constraint $((x_1, \ldots, x_l), A^l)$ for all pairwise different variables $x_1, \ldots, x_l \in \Var$ to satisfy the first condition and then remove tuples (orbits of tuples) from  the relations in constraints in the instance as long as the second condition is not satisfied. It is widely known and easy to prove that an instance $\instanceJ$ of the CSP obtained by the described algorithm has the same set of solutions as $\instance$: they are equivalent. In particular,  if $\instanceJ$ is trivial, then $\instance$ has no  solutions.
Under a natural assumption that $\structB$ contains all at most $l$-ary relations $pp$-definable in $\structB$, we have that $\instanceJ$ is an instance of $\csp(\structB)$. From now on  this assumption will be in effect.

\begin{definition}
\label{def:relwidth}
A relational structure $\structB$ has \emph{relational width $(k,l)$} if every $(k,l)$-minimal instance $\instance$ of $\csp(\structB)$ has a solution iff it is non-trivial.

A relational structure $\structB$ has \emph{bounded relational width} if it has relational width $(k,l)$ for some natural numbers $k \leq l$.
\end{definition}


\section{Implicationally Uniform Relational Clones Imply Bounded Width}
\label{sect:ImpUniform}

In this section we deal with $\structB$ whose relational clones have a very particular property.

\begin{definition}
\label{def:impuniform}
We say that a relational clone  is \emph{implicationally uniform} if
for every pair of complementary implications $R_1: \subsetA \rightarrow \subsetB$ and $R_2: \subsetB \rightarrow \subsetA$ in it we have that  
$\subsetA = \subsetB$. Otherwise, we say that a relational clone is 
\emph{implicationally non-uniform}.
\end{definition}

In this section, we show that if a relational clone of $\structB$ is implicationally uniform, then $\csp(\structB)$ is solvable by establishing $(2, \maxboundA)$-minimality where $\maxboundA$ is the size of the largest forbidden substructure in $\forbA$ but not smaller than $3$. We consider instances over variables $\{ v_1, \ldots, v_n\}$.
Since instances $\instance$ of interest are $(2, \maxboundA)$-minimal, the projection of every constraint to variables $v_i, v_j$ with $i,j \in [n]$ is the same binary relation. We denote it by $\instance_{i,j}$. The next important notion that we are going to use is  a (directed) graph of an instance.

\begin{definition}
\label{def:graphinstance}
Let $\instance$ be an instance of $\csp(\structB)$ over variables $\{ v_1, \ldots, v_n \}$ and a
first-order expansion $\structB$ of a finitely bounded homogeneous symmetric binary core $\structA$. A (directed) graph  $\graphinstance$ of an instance $\instance$ 
consists of vertices of the form $((v_i, v_j), \subsetA)$  such that 
$\subsetA \subsetneq \instance_{i,j}$ has a pp-definition in $\structB$.

There is an arc between 
$((v_i, v_j), \subsetA)$ and $((v_k, v_l), \subsetB)$ in $\graphinstance$ if there is a relation $R$ in the relational clone of $\structB$ such that
\begin{itemize}
\item  $\Pi_{1,2}(R) = \instance_{i,j}$, 
$\Pi_{-2,-1}(R) = \instance_{k,l}$ and 
\item $R$ is a $(\subsetA \rightarrow \subsetB)$-implication. 
\end{itemize}
\end{definition}

We are now ready to prove that the minimality algorithm solves instances of $\csp(\structB)$ such that the relational clone of $\structB$ is implicationally uniform clones.

\begin{theorem}
\label{thm:impuniform}
    Let $\structB$ be a first-order expansion of a finitely bounded homogeneous symmetric binary core $\structA$ such that the relational clone of $\structB$ is implicationally uniform and 
    $\instance$  a non-trivial $(2,\maxboundA)$-minimal instance of $\csp(\structB)$ where $\maxboundA$ is the maximal size of a structure in $\forbA$ but not smaller than $3$. Then $\instance$ has a solution. 
\end{theorem}

\begin{proof}
We will measure the size of an instance  $\instance$ of $\csp(\structB)$ by the sum of the number of orbits across all constraints, i.e.,
$$\text{size}(\instance) = \sum_{(\overline{x}, R) \in \instance} \text{number of orbits}(R)$$

Suppose now that the theorem does not hold and take a minimal in size instance of $\csp(\structB)$ that is non-trivial $(2, \maxboundA)$-minimal and has no solution. We start with a case where every $\instance_{i,j}$ with $i,j \in [n]$ contains only one orbit. 
In this case we construct a finite structure $\structC$ over the domain consisting of variables in $\instance$ so that $(v_i, v_j) \in \orbitO^{\structC}$ if and only if  $\instance_{i,j} = \orbitO$.
Since $\maxboundA \geq 3$ we have that $\structC$ respects the transitivity of $=$, i.e. we have $(v_{i_k}, v_{i_m}) \in =^{\structC}$
whenever $(v_{i_k}, v_{i_l}) \in =^{\structC}$ and 
$(v_{i_l}, v_{i_m})  \in =^{\structC}$. In fact, $=^{\structC}$ is an equivalence relation on $\{ v_1, \ldots, v_n \}$. 
Thus consider a quotient structure $\structD$ which is $\structC /  
=^{\structC}$ and whose domain is the set of equivalence classes of 
$=^{\structC}$. We have $([v_i]_{=^{\structC}}, [v_j]_{=^{\structC}}) \in \orbitO^{\structC}$
if and only if $\instance_{i,j} = \orbitO$. 
Since $\instance$ is $(2,\maxboundA)$-minimal there is an embedding 
of $\structD$ into $\structA$ and a homomorphism of $\structC$ into
$\structA$. The latter is also a solution to $\instance$ which was to be proved. 

The previous case represents the situation where $\graphinstance$ is empty. Now, we move to the case where at least one $\instance_{i,j}$ with $i \neq j$ in $[n]$ is not an orbital. Since every orbital is 
pp-definable in $\structA$ and thereby in $\structB$, and $\instance_{i,j}$ contains at least two orbitals, from now on the graph $\graphinstance$ is non-empty. 
We will look at  maximal strongly connected components $\maximal$
of $\graphinstance$, i.e. strongly connected components such that every arc originating in $\maximal$ ends up in $\maximal$. 
For the sake of simplicity we simply say that  $\maximal$ is a maximal component of $\graphinstance$. We start with a simple observation. (All omitted proofs may be found in the appendix.)

\begin{observation}
\label{obs:maxcompuniform}
Let $\maximal$ be a maximal component of $\graphinstance$, then for all pairs of vertices $((v_i,v_j), \subsetA)$ and
$((v_k,v_l), \subsetB)$ in $\maximal$ we have $\subsetA = \subsetB$. 
\end{observation}

\begin{proof}

Assume on the contrary that there is a cycle in $\maximal$ containing both
$((v_i,v_j), \subsetA)$ and
$((v_k,v_l), \subsetB)$ with $\subsetA \neq \subsetB$. 
It implies that there is a 
chain of implications $R_1: \subsetA_1 \rightarrow \subsetA_2, R_2: \subsetA_2 \rightarrow \subsetA_3,  \ldots, R_a: \subsetA_{a} \rightarrow \subsetB, R_{a+1}: \subsetB \rightarrow \subsetB_1,  R_{a+2}: \subsetB_1 \rightarrow \subsetB_2,  
\ldots, R_{a+b}: \subsetB_{a+b} \rightarrow \subsetA$
satisfying the conditions in Definition~\ref{def:graphinstance}.
In particular all these relations are pp-definable in $\structB$.
By use of $\circ$-composition in Definition~\ref{def:RcircR}, we have that $((R_1 \circ R_2) \circ \cdots \circ R_a)$
is a $(\subsetA \rightarrow \subsetB)$-implication 
and that $((R_{a+1} \circ R_{a+2}) \circ \cdots \circ R_{a+b})$ is a
 $(\subsetB \rightarrow \subsetA)$-implication. Clearly $R_1, R_2$ are complementary. The fact that $\subsetA \neq \subsetB$ contradicts the fact that the relational clone of $\structB$ is implicationally uniform. 
\end{proof}

Thus, from now on we assume that $\graphinstance$ contains a maximal strongly connected component $\maximal$ in which all vertices are of the form $((v_i, v_j), \subsetA)$ with the same $\subsetA$. We reduce the size of an $\instance$ so that  every $R$ in a constraint $((x_1, \ldots, x_k), R)$ in $\instance$ is replaced by 
$$ R(x_1, \ldots, x_k) \wedge \bigwedge_{((x_i, x_j), \subsetA) \in \maximal} \subsetA(x_i,x_j).$$
The new instance will be denoted by $\instanceJ$. 
In order to show that it is non-trivial and $(2, \maxboundA)$-minimal consider the following observation, which directly follows from the definition of $\graphinstance$.

\begin{observation}
\label{obs:RwithAs}
Let $R$ be a $k$-ary relation in a relational clone of $\structB$ and $R(v_{i_1}, \ldots, v_{i_k})$ an atomic formula over variables of the instance $\instance$ such that for all $l,m \in [k]$ the projection of  $R(v_{i_1}, \ldots, v_{i_k})$ to $v_{i_l}, v_{i_m}$ is $\instance_{i_l, i_m}$.
Then for all $c>0$ and $a_1, b_1, \ldots, a_c, b_c, l,k \in [m]$
satisfying $((v_{i_{a_j}}, v_{i_{b_j}}), \subsetA) \in \maximal$ 
for all $j \in [c]$
we have that the
 projection of  
$$R'(v_{i_1}, \ldots, v_{i_k}) \equiv \left( R(v_{i_1}, \ldots, v_{i_k}) \wedge \bigwedge_{j \in [c]} \subsetA(v_{i_{a_{j}}}, v_{i_{b_{j}}}) \right)$$ 
to any $v_{i_l}, v_{i_m}$ with $l,m \in [k]$ is either 
\begin{itemize}
\item $\instance_{i_l, i_m}$ or 
\item $\subsetA$ and then $((v_{i_l}, v_{i_m}), \subsetA) \in \maximal$. 
\end{itemize}
\end{observation}

\begin{proof}
Assume the contrary and take a minimal $c$ such that
there are $v_{i_l}, v_{i_m}$ with $l,m \in [k]$ and such that the projection of  
$R'(v_{i_1}, \ldots, v_{i_k}) $ to $v_{i_l}, v_{i_m}$ is 
$\subsetB \notin \{ \subsetA, \instance_{i_l, i_m}, \emptyset \}$.
Observe that since $c$ is minimal, we can assume that $\subsetB \neq \emptyset$. 
Then let $e \in [c]$ be such that for $I = [c] \setminus \{e\}$ we have that the projection of 
$$R''(v_{i_1}, \ldots, v_{i_k}) \equiv \left( R(v_{i_1}, \ldots, v_{i_k}) \wedge \bigwedge_{j \in I} \subsetA(v_{i_{a_{j}}}, v_{i_{b_{j}}}) \right)$$ 
to $v_{i_{a_e}}, v_{i_{b_e}}, v_{i_l}, v_{i_m}$ is a $(\subsetA \rightarrow \subsetB)$--implication for some $\subsetB \subsetneq \instance_{i_l, i_m}$. 
Since $\subsetA$ and $R$ are pp-definable in $\structB$, by the definition of $\graphinstance$, it follows that $\maximal$ contains  $((v_{i_l}, v_{i_m}), \subsetB)$ which contradicts Observation~\ref{obs:maxcompuniform}.

The other thing that could happen is when the projection of $R'(v_{i_1}, \ldots, v_{i_k})$ to $v_{i_l}, v_{i_m}$ is 
$\subsetA$ and $((v_{i_l}, v_{i_m}), \subsetA) \notin \maximal$.
But then again by the definition of $\graphinstance$ the vertex
$((v_{i_l}, v_{i_m}), \subsetA)$ should have been in $\maximal$.
It completes the proof of the observation. 
\end{proof}

By the observation above it follows that for every constraint $((v_{i_1},\ldots, v_{i_k}), R) \in \instance$ a projection of 
$$R(v_{i_1},\ldots, v_{i_k}) \wedge \bigwedge_{((v_{i_l}, v_{i_m}), \subsetA) \in \maximal} \subsetA(v_{i_l},v_{i_m})$$
to $v_{i_l}, v_{i_m}$ for any $l,m \in [k]$ is either $\instance_{i_l, i_m}$ or
$\subsetA$ and the latter case holds only when $((v_{i_l}, v_{i_m}), \subsetA) \in \maximal$. 
It follows that $\instanceJ$ is non-trivial and $(2,\maxboundA)$-minimal. 
Furthermore, if $\instance$ does not have a solution, then $\instanceJ$ does not have it either. Since $\instanceJ$ is of a smaller size than $\instance$ we have a contradiction with the minimality of $\instance$.
\end{proof}


\section{Implicationally Non-Uniform Clones Are Not Preserved By Chains of Quasi Directed  J\'{o}nsson Operations}
\label{sect:ImpNonUniform}

In this section we show that whenever a relational clone of $\structB$ is implicationally non-uniform, i.e. it contains complementary relations $R_1:\subsetA \rightarrow \subsetB$ and $R_2: \subsetB \rightarrow \subsetA$ with $\subsetA \neq \subsetB$, then  $\structB$ is not preserved by any chain of quasi directed J\'{o}nsson operations.

We start with a definition of a directed bipartite graph $\bipartite_{R_1, R_2}$ which reflects the structure of 
$R_1, R_2$ which agree on projections. We use it mainly for $R_1, R_2$ which are complementary implications.

If it is not stated otherwise we will assume that both $R_1$ and $R_2$
are quaternary.
To this end, for every binary relation $\subsetA$ first-order definable in $\structA$ we set  
  \begin{itemize}
  \item $\vertices_L(\subsetA) := \{ \orbitO_L \mid \orbitO \textrm{ is an orbital contained in  } \subsetA \}$, and
 \item $\vertices_R(\mathcal{A}):=\{  \orbitO_R \mid \orbitO \textrm{ is an orbital contained in  } \subsetA \}$.
\end{itemize}

\begin{definition}
\label{def:bipartite}
Let $R_1, R_2$
 be two relations that agree on projections. We define  $\bipartite_{R_1, R_2}$ to be a bipartite digraph over vertices
$\vertices_L(\Pi_{1,2}(R_1)) \cup \vertices_R(\Pi_{1,2}(R_2)$ where we have two kinds of arcs:
\begin{itemize}
\item $(\orbitO_L, \orbitP_R) \in \vertices_L(\Pi_{1,2}(R_1)) \times \vertices_R(\Pi_{1,2}(R_2))$ if 
the relation $R_1$ contains a $(\orbitO, \ldots, \orbitP)$-tuple,
\item $(\orbitO_R, \orbitP_L)   \in \vertices_R(\Pi_{1,2}(R_2)) \times \vertices_L(\Pi_{1,2}(R_1))$ if 
the relation $R_2$ contains   a
$(\orbitO, \ldots, \orbitP)$-tuple.
\end{itemize}
\end{definition}

A strongly connected component of $\bipartite_{R_1, R_2}$ 
is simply called a component. A component is non-trivial if it contains more than one vertex.
A component $\mathcal{N}$
is maximal if all outgoing edges end up in $\mathcal{N}$ and it is minimal if  all ingoing edges originate in $\mathcal{N}$.

For the sake of simplicity we say that a subgraph of $\bipartite_{R_1, R_2}$ induced by
$(\vertices_L(\subsetC) \cup \vertices_R(\subsetD))$ for some binary relations $\subsetC \subseteq \Pi_{1,2}(R_1)$ and $\subsetD \subseteq \Pi_{1,2}(R_2)$ is a $(\subsetC, \subsetD)$-subgraph of $\bipartite_{R_1, R_2}$. 
In particular, a component $(\vertices_L(\subsetC) \cup \vertices_R(\subsetD))$ of $\bipartite_{R_1, R_2}$ induces a $(\subsetC, \subsetD)$--subgraph and in this case we say also that that component is  a $(\subsetC, \subsetD)$---component.
 We say that a $(\subsetC, \subsetD)$-component in $\bipartite_{R_1, R_2}$ contains a tuple $t$
if it is a $(\orbitC, \ldots, \orbitD)$-tuple in $R_1$ or a $(\orbitD, \ldots, \orbitC)$-tuple in $R_2$ with $\orbitC \subseteq \subsetC$ and $\orbitD \subseteq \subsetD$.

\begin{observation}
\label{obs:minmaxcomponents}
Let $R_1: \mathcal{A} \rightarrow \mathcal{B}$ and $R_2: \mathcal{B} \rightarrow \mathcal{A}$ be complementary. Then a $(\subsetA, \subsetB)$ subgraph 
of
$\bipartite_{R_1, R_2}$ has both a  minimal 
and a
maximal component. 
Furthermore, a $(\Pi_{1,2}(R_1) \setminus \subsetA, \Pi_{3,4} \setminus \subsetB)$-subgraph of $\bipartite_{R_1, R_2}$ has a minimal component.
\end{observation}

\begin{proof}
Since $\subsetA + R_1 = \subsetB$ 
and $\subsetB + R_2 = \subsetA$, we have that  for any vertex $\orbitA_L$ in $\vertices_L(\subsetA)$ there exists 
$\orbitB_R$ in $\vertices_B(\subsetA)$ 
such that $R_1$ contains an $(\orbitA, \ldots, \orbitB)$-tuple 
and that for any 
$\orbitB_R$ in $\vertices_B(\subsetA)$
we have $\orbitA_L$ in $\vertices_L(\subsetA)$ such that 
a $(\orbitB, \ldots, \orbitA)$-tuple is in $R_2$.
It follows that by an appropriately long walk we can reach a component from where there is no outgoing arc. It is a maximal component in the $(\subsetA, \subsetB)$--subgraph of $\bipartite_{R_1, R_2}$. For the minimal component in  
the $(\subsetA, \subsetB)$--subgraph of $\bipartite_{R_1, R_2}$ notice that for any $\orbitA_L$ in $\vertices_L(\subsetA)$ there exists 
$\orbitB_R$ in $\vertices_R(\subsetB)$ such that 
a $(\orbitB, \ldots, \orbitA)$-tuple is in $R_2$
and that for any 
$\orbitB_R$ in $\vertices_R(\subsetB)$ there exists 
$\orbitA_L$ in $\vertices_L(\subsetA)$ such that 
a $(\orbitA, \ldots, \orbitB)$-tuple is in $R_1$.
Thus by walking from any vertex in $(\subsetA, \subsetB)$--subgraph backwards but inside the subgraph we reach a minimal component in 
$(\subsetA, \subsetB)$.
For the minimal component in a $(\Pi_{1,2}(R_1) \setminus \subsetA, \Pi_{3,4} \setminus \subsetB)$--subgraph of $\bipartite_{R_1, R_2}$ we do the same thing just start with any vertex in there and go against the arrows. 
\end{proof}

We will now define a $\circ$-composition of two quaternary relations, another slightly different $\bowtie$-composition is defined later on.

\begin{definition}
\label{def:RcircR}
Let $R_1, R_2$ be two quaternary relations such that $\pi_{-2,-1}(R_1) = \pi_{1,2}(R_2)$. We define a $\circ$-\emph{composition} $R_1 \circ R_2$ of $R_1$ and $R_2$ to be the relation defined by the formula
$$\exists y \exists z~R_1(x_1,  x_2, y, z) \wedge R_2(y, z, x_3,  x_4).$$
\end{definition}

We will write $(R_1 \circ R_2)^n$ as a shorthand for the expression $(( \cdots (((R_1 \circ R_2) \circ R_1) \circ R_2) \circ \cdots \circ R_1) \circ R_2)$ where both $R_1$ and $R_2$ occur $n$ times. 
\begin{observation}
\label{obs:circbipartite}
  Let $R_1: \subsetA \rightarrow \subsetB$ and $R_2: \subsetB \rightarrow \subsetA$ be quaternary complementary implications. Then, for all $n \geq 1$ we have both of the following: 
  \begin{itemize}
      \item $S_{2n} \equiv (R_1 \circ R_2)^n$ contains an $(\orbitO,\ldots, \orbitP)$-tuple
      iff there is a path in $\bipartite_{R_1, R_2}$ of length $2n$ from $\orbitO_L$ to $\orbitP_L$;
      \item $S_{2n+1} \equiv (R_1 \circ R_2)^n \circ R_1$ contains an $(\orbitO, \ldots, \orbitP)$-tuple
      iff there is a path in $\bipartite_{R_1, R_2}$ of length $2n+1$ from $\orbitO_L$ to $\orbitP_R$. 
  \end{itemize}
\end{observation}

\begin{proof}
    We prove the observation by induction on $n \geq 1$.
    Notice that if $R_{1}$ contains an $(\orbitO, \ldots, \orbitA)$-tuple and $R_2$ contains an
    $(\orbitA, \ldots, \orbitP)$-tuple for some orbital $\orbitA$, then by the homogeneity of $\structA$, the relation 
    $S_{2}$ has an $(\orbitO, \ldots, \orbitP)$-tuple. On the other hand if there are no such tuples, by the definition of $S_2$, there is not an $(\orbitO, \ldots, \orbitP)$-tuple in $S_2$.

    In the induction step consider $S_{2n+1} \equiv S_{2n} \circ R_1$. If there is an $(\orbitO,\ldots, \orbitP)$-tuple in $S_{2n+1}$, then for some orbital $\orbitA$, by the definition of this relation, we have an $(\orbitO, \ldots, \orbitA)$-tuple in $S_{2n}$ and
     an
    $(\orbitA, \ldots, \orbitP)$-tuple in $R_1$. By the induction hypothesis, there is a path of length $2n$ 
    from $\orbitO_L$ to $\orbitA_L$ in $\bipartite_{R_1, R_2}$, and hence a path of length $2n+1$ from $\orbitO_L$ to $\orbitP_R$.
    On the other hand, if there is a path from $\orbitO_L$ to $\orbitP_R$ in $\bipartite_{R_1, R_2}$ of length $2n+1$, then there exists some $\orbitA$ so that there is a path of length $2n$ from $\orbitO_L$ to $\orbitA_L$ in $\bipartite_{R_1, R_2}$ and a $(\orbitA, \ldots, \orbitP)$-tuple in $R_1$.
    By the induction hypothesis, there is an $(\orbitO, \ldots, \orbitA)$-tuple in $S_{2n}$. By the homogeneity of $\structA$, there is an $(\orbitO, \ldots, \orbitP)$-tuple in $S_{2n+1}$. 
    The proof for $S_{2n}$ with $n > 1$ is analogous. We just replace $R_1$ with $R_2$ and $S_{2n+1}$ with $S_{2n}$ in the proof above.  
    It completes the proof of the observation.
\end{proof}

The next step is to extend the above observation so that it specifies what kind of a tuple we can expect, for that we need the following definition. 

\begin{definition}
\label{tuples:sorts}
A quaternary tuple $t$ is
\begin{itemize}
\item \emph{degenerated} if $t[1] = t[4]$ and $t[2] = t[3]$,
\item \emph{essentially ternary} if $t[2] = t[3]$ and $t[1] \neq t[4]$,
\item \emph{essentially quaternary} if $t[i] \neq t[j]$ whenever $i \in \{ 1,2\}$ and $j \in \{ 3,4 \}$, 
\item \emph{partially-free} if  $t[1,4] \in \orbitN$.
\item \emph{fully-free} if $(t[i], t[j]) \in \orbitN$ whenever $i \in \{ 1,2\}$ and $j \in \{ 3,4 \}$.
\end{itemize}
\end{definition}

\noindent
We now prove observation that tells us what kind of tuples we can expect in a $\circ$-composition of two relations.  

\begin{observation}
\label{obs:tuplescomposition}
Let $R_1 : \subsetA \rightarrow \subsetB$ and $R_2: \subsetB \rightarrow \subsetC$ be such that $R_1$ contains a $(\orbitA, \ldots, \orbitB)$-tuple $t_1$ and $R_2$ contains a $(\orbitB,\ldots, \orbitC)$-tuple $t_2$. Then $R_3 := R_1 \circ R_2$ is a 
$(\subsetA \rightarrow \subsetC)$-implication containing a $(\orbitA, \ldots, \orbitC)$-tuple $t_3$ which is
\begin{itemize}
\item essentially ternary if both $t_1$ and $t_2$ are essentially ternary, $\orbitA, \orbitC$ anti-reflexive and $\orbitB$ is $=$, 
\item essentially quaternary if both $t_1, t_2$ are essentially ternary and all $\orbitA, \orbitB, \orbitC$ are anti-reflexive,
\item fully-free if both $t_1, t_2$ are essentially quaternary,
\item partially-free if $t_1$ is essentially quaternary and  $t_2$ essentially ternary and $\orbitC$ is not $=$ or 
$t_1$ is essentially ternary and  $t_2$ essentially quaternary and $\orbitA$ is not $=$,
\item essentially quaternary (fully-free) if at least one of $t_1, t_2$ is essentially quaternary (fully-free) 
\item non-degenerated if at least one of them is non-degenerated.
\end{itemize}
\end{observation}

\begin{proof}
It is straightforward that $R_3$ is a $(\subsetA \rightarrow \subsetC)$-implication. The rest of the observation we prove case by case. 

The age of $\structA$ has free amalgamation and therefore there are $a, b, c \in A$ such that $(a,b) \in \orbitA, (b,c) \in \orbitB$ and $(a,c) \in \orbitN$. The first item follows.

For the second case we consider $a,b,c,d \in \orbitA$ such that
$(a,b,b,c)$ is isomorphic with $t_1$ and $(b,c,c,d)$ with $t_2$.
Since the age of $\structA$ has free amalgamation, we may assume that  $(a,d) \in \orbitN$. Observe that $(a,b,c,d)$ is in $R_3$ and is essentially quaternary. 

For the third case we need $a,b,c,d,e,f \in A$ such that
$(a,b,c,d)$ is isomorphic with $t_1$ and $(c,d,e,f)$ with $t_2$. Again using the fact that 
the age of $\structA$ has free  amalgamation, we may assume that  $(i,j) \in \orbitN$ whenever $i \in \{ a,b \}$ and $j \in \{ e,f \}$. Observe that $(a,b,e,f)$ is in $R_3$ and is fully free.

For the fourth case we only look at the first part. To this end we consider $a,b,c,d,e$ such that $(a,b,c,d)$ is isomorphic to $t_1$ and $(c,d,e)$ to $t_2$. Since $a \neq d$ and $d \neq e$ we may assume that $(a,e) \in \orbitN$. It follows that $(a,b,d,e)$ is a partially free and is in $\orbitN$. 

The proof of the fifth and the sixth case is straightforward and very similar to the proof of the preceding cases.  
\end{proof}

We now define a self-complementary implication.

\begin{definition}
\label{def:selfcomp}
We say that an implication $R: \subsetA \rightarrow \subsetA$ is \emph{self-complementary} if $R$ and $R$ are complementary and every non-trivial component of $\bipartite_{R_1, R_2}$ is a $(\subsetC, \subsetC)$-component.     
\end{definition}


The following will be used throughout the paper.

\begin{observation}
\label{obs:selfcomplementary}
Let $R_1: \mathcal{A} \rightarrow \mathcal{B}$ and $R_2: \mathcal{B} \rightarrow \mathcal{A}$ be complementary implications. Then $R \equiv R_1 \circ R_2$ is a self-complementary  $(\subsetA \rightarrow \subsetA)$-implication. 
\end{observation}
\begin{proof}

The relation $R$ is clearly a $(\subsetA \rightarrow \subsetA)$-implication. Since every vertex in $\bipartite_{R_1, R_2}$ has an incoming edge we have that $\Pi_{3,4} (R) = \Pi_{1,2}(R)$, and hence $R,R$ are complementary.  

Observe now that every non-trivial maximal component in considered bipartite digraphs corresponds to a maximal in the number of pairwise different vertices closed walk.
From a closed walk in $\bipartite_{R_1, R_2}$ corresponding to a maximal $(\subsetC, \subsetD)$-component one can get a closed walk in $\bipartite_{R,R}$ corresponding to a component $(\subsetC, \subsetC)$ in $\bipartite_{R,R}$ by contracting all two arcs with a vertex from $\vertices_R(\subsetD)$ in the middle into one arc. Suppose now that it is not a maximal component, i.e. there is a $(\subsetC', \subsetC'')$-component in $\bipartite_{R,R}$ containing the $(\subsetC, \subsetC)$-component. But then from a corresponding walk one can create a closed walk in $\bipartite_{R_1, R_2}$ corresponding to a component containing vertices  $(\vertices_L(\subsetC') \cup \vertices_L(\subsetC'') \cup \vertices_R(\subsetD))$. It contradicts the maximality of 
the $(\subsetC, \subsetD)$-component in $\bipartite_{R_1, R_2}$. It follows that $R$ is a self-complementary implication.
\end{proof}

A non-trivial component in $\bipartite_{R_1, R_2}$ for two $R_1, R_2$ which agree on projections is \emph{degenerated} if all tuples in it are degenerated.
Observe that a non-trivial degenerated component contains exactly two  $(\orbitO,\orbitO,=,=, \orbitO,\orbitO)$-tuples for some orbital $\orbitO$ and therefore we will call it an $\orbitO$-degenerated component.

In the proof announced at the beginning of this section we distiguish two cases: either there is a non-trivial component in $\bipartite_{R_1, R_2}$ which is non-degenerated or all non-trivial components in $\bipartite_{R_1, R_2}$ are degenerated. We take care of the former case in Section~\ref{sect:nondegen} and of the latter in Section~\ref{sect:alldegen}.

\subsection{A non-trivial non-degenerated component} 
\label{sect:nondegen}

We start from a proposition which reduces the situation we take
care of in this section to two simple cases.

\begin{proposition}
\label{prop:nondegen_reduce}
Let  $\structB$ be a first-order expansion of a finitely bounded homogeneous symmetric binary core whose age has free amalgamation.
If $\structB$ pp-defines complementary implications $R_1:\subsetA \rightarrow \subsetB, R_2:\subsetB \rightarrow \subsetA$  such that $\bipartite_{R_1, R_2}$ contains a non-trivial non-degenerated component. Then $R$ pp-defines a self-complementary $R: \subsetA \rightarrow \subsetA$ such that $R$ contains
 \begin{itemize}
 \item a $(\orbitA, \orbitN, \orbitN, \orbitN,\orbitN, \orbitA)$-tuple for some $\orbitA \subseteq \subsetA$ and
 \item a $(\orbitB, \orbitN, \orbitN, \orbitN,\orbitN, \orbitB)$-tuple or a $(\orbitB, \orbitB, =, =,\orbitB, \orbitB)$-tuple for some $\orbitB \nsubseteq \subsetA$. 
 \end{itemize}
\end{proposition}

\begin{proof}
Let $t$ be a non-degenerated $(\orbitA, \ldots, \orbitC)$-tuple in a nontrivial $(\subsetA,\subsetA)$-component of $\bipartite_{R_1, R_2}$. Without loss of generality we may assume that either $t$ is essentially quaternary or $\orbitA$ is anti-reflexive. Indeed, if $t[2] \neq t[3]$ and $t[1] = t[4]$ we replace $R_1$ with $R_1(x_2,x_1, x_4, x_3)$,  if $\orbitC$ is $=$, then we swap $R_1$ with $R_2$ and if $\orbitA, \orbitC$ are both $=$ and $t$ is essentially ternary, then $t$ is degenerated which contradicts the asumption.

By Observation~\ref{obs:circbipartite}, it follows 
that there is $k$ such that  $R'_2 \equiv R_2 \circ (R_1 \circ R_2)^k$ 
contains  a $(\orbitB, \ldots, \orbitA)$-tuple and that 
$R \equiv R_1 \circ R'_2$ is a self-complementary implication that contains a non-degenerated  $(\orbitA, \ldots, \orbitA)$-tuple. The last tuple is non-degenerated by Observation~\ref{obs:tuplescomposition}. 
By the same observation, it follows that $R' = (R \circ R)^2$ contains an essentially quaternary  $(\orbitA, \ldots, \orbitA)$-tuple and that 
$S := (R' \circ R')^2$ has a fully-free  $(\orbitA, \ldots, \orbitA)$-tuple, i.e., a $(\orbitA, \orbitN, \orbitN, \orbitN, \orbitN, \orbitA)$-tuple.
The relation $S$ is clearly a self-complementary $(\subsetA \rightarrow \subsetA)$-implication by Observation~\ref{obs:selfcomplementary}.

By Observation~\ref{obs:minmaxcomponents}, there is a $(\subsetB, \subsetB)$-component in $\bipartite_{R_1, R_2}$ with $\subsetA \cap \subsetB = \emptyset$.
It follows that  there exists $l$ such that $S':=S^l$ contains both an
$(\orbitA, \orbitN, \orbitN, \orbitA)$-tuple and a $(\orbitB, \ldots, \orbitB)$-tuple $s$ for some $\orbitB$. If the latter 
tuple is not a $(\orbitB, \orbitB, =, =,\orbitB, \orbitB)$-tuple and $\orbitB$
is not $=$, then as above we show that $T = (((S'\circ S')^{2})^2$ has  
a $(\orbitB, \orbitN, \orbitN, \orbitN, \orbitN, \orbitB)$-tuple. If $\orbitB$ is $=$, then the non-degenerated $s$ has to be essentially quaternary. Then, again, by Observation~\ref{obs:tuplescomposition}, $T$ has a   $(\orbitB, \orbitN, \orbitN, \orbitN, \orbitN, \orbitB)$-tuple. $T$ clearly also contain a $(\orbitA, \orbitN, \orbitN, \orbitN, \orbitN, \orbitA)$-tuple and is a self-complementary implication.

Finally, observe that  $\bipartite_{T', T'}$ where $$T'(x_1, x_2, x_3, x_4) \equiv T(x_1,x_2, x_3, x_4) \wedge T(x_4,x_3, x_2, x_1)$$ contains a $(\orbitA, \ldots \orbitB)$-tuple only if $T$ contains that tuple and additionally a $(\orbitB, \ldots \orbitA)$-tuple. Since all components in $\bipartite_{T,T}$ are $(\subsetC, \subsetC)$-components the orbitals $\orbitA$ and $\orbitB$
cannot come from two different components. Indeed, if $T$ contains a 
$(\orbitA, \ldots, \orbitB)$-tuple where $\orbitB_R$ is in a component above
the component of $\orbitA_L$, then by the shape of the components it cannot have a $(\orbitB, \ldots, \orbitA)$-tuple. It follows that $\bipartite_{T', T'}$ is a union of disjoint non-trival  $(\subsetA, \subsetA)$-components each of which is a subset of some non-trivial component in $\bipartite_{T, T}$. In particular, we have that $T'$ is a $(\subsetA \rightarrow \subsetA)$-implication with $\orbitA \subseteq \subsetA$. 
\end{proof}

It is a trivial thing that $=$ is transitive. We also observe that for any orbital $\orbitO$ and for all elements $a,b,c \in A$ such that $(a,b) \in \orbitO$ and $(b=c)$; or
$a = b $ and $(b,c) \in \orbitO$
we have $(a,c) \in \orbitO$. The latter property of $=$  will be called \emph{semi-transitivity}. Obviously, semi-transitivity implies transitivity.

We will now show that in both cases in the lemma above $\structB$
is not preserved by a chain of quasi directed J\'{o}nsson operations.
Each time we will use the following definition.  
\begin{definition}
\label{def:UVconstantAB}
Let $\structA$ be a finitely bounded homogeneous symmetric binary core whose age has free amalgamation, $U,V  \subseteq [3]$ and $\orbitM, \orbitO, \orbitP \in \{ \orbitA, \orbitB \}$ for some orbitals $\orbitA, \orbitB$.
We say that $(u,v)$ in $A^3$ is a $(U,V)$-constant pair of triples satisfying $\orbitM \orbitO \orbitP(u,v)$ if all of the following hold.
\begin{enumerate}
\item \label{UVconstant:one} $(u[1], v[1]) \in \orbitM$, $(u[2], v[2]) \in \orbitO$, $(u[3], v[3]) \in \orbitP$.
\item \label{UVconstant:two} $u[i] = u[j]$ iff $i,j \in U$ and $v[i] = v[j]$ iff $i,j \in V$.
\item \label{UVconstant:three} For all other $a,b \in \{ u[1], u[2], u[3], v[1], v[2], v[3] \}$ for which the above conditions and the semi-transitivity of equality do not determine the orbital, we have $(a,b) \in \orbitN$.
\end{enumerate} 
\end{definition}

Obviously not for all $U,V \subseteq [3]$ and $\orbitM, \orbitO, \orbitP$
there are $u, v$ satisfying the conditions above. Therefore anytime we use that definition  we make sure the appropriate $u,v$ exist.


\begin{lemma}
\label{lem:bnbnoJonsson}
Let $\structB$ be a  first-order expansion of a finitely bounded homogeneous symmetric binary core $\structA$ whose age has free amalgamation. If $\structB$ pp-defines a quaternary self-complementary relation $R: \subsetA \rightarrow \subsetA$.
which contains  
\begin{itemize}
\item   a $(\orbitA, \orbitN, \orbitN, \orbitN, \orbitN, \orbitA)$-tuple and 
\item  a $(\orbitB, \orbitN, \orbitN, \orbitN, \orbitN, \orbitB)$-tuple, 
\end{itemize}
for two different orbitals $\orbitA \subseteq \subsetA$ and $\orbitB \nsubseteq \subsetA$, 
then $R$ is not preserved by any chain of quasi directed   J\'{o}nsson operations.   
\end{lemma}

\begin{proof}
Suppose that there is a chain $(D_1, \ldots,D_n)$ of quasi directed J\'{o}nsson operations preserving $\structB$ and $R$. We will show that this assumption contradicts the fact that $\Pol(\structB)$ preserves $\orbitB$ which should be preserved since $structB$ contains all orbitals w.r.t $\Aut(\structA)$.

At various steps of the proof we simultaneously consider three different cases: 
\begin{enumerate}
\item both $\orbitA$ and $\orbitB$ are anti-reflexive,
\item only $\orbitA$ is anti-reflexive,
\item only $\orbitB$ is anti-reflexive.
\end{enumerate}

In the  following claim we consider two different types of $(U, V)$-constant pairs of triples. Inspect Figure~\ref{fig:bnbextension} and observe that $t,s$ is a  
$([3],[2])$-constant pair of triples  in $A$
satisfying $\orbitA \orbitA \orbitB(u,v)$
and that $u,v$ is a 
$([3],\{2,3\})$-constant pair of triples $(u,v)$ in $A$ 
satisfying $\orbitA\orbitB\orbitB(u,v)$.
According to Definition~\ref{def:UVconstantAB}  we have that 
 $\orbitO$ is
\begin{enumerate}
\item   $\orbitN$  if  $\orbitA$ and $\orbitC$ are anti-reflexive,
\item  $\orbitB$ if $\orbitA$ is $=$ and $\orbitB$ is anti-reflexive, 
\item $\orbitA$ if $\orbitC$ is $=$ and $\orbitA$ is anti-reflexive.
\end{enumerate}

We are now ready to formulate and prove the claim whose proof is a large part of the proof of the lemma.

\begin{claim}
\label{claim:bnbnoJonsson}
For all $i \in [n]$ we have both of the following.
\begin{itemize}
\item  For all $([3],[2])$-constant pairs of triples $(u,v)$ in $A$ 
satisfying $\orbitA \orbitA \orbitB(u,v)$ we have that $(D_i(u), D_i(v)) \in \subsetA$. 
\item For all $([3],\{2,3\})$-constant pairs of triples $(u,v)$ in $A$ 
satisfying $\orbitA\orbitB\orbitB(u,v)$ we have that  $(D_i(u), D_i (v)) \in \subsetA$. 
\end{itemize}
\end{claim}

\begin{proof}
We prove the claim by the induction on $i \in [n]$. 

\textbf{(BASE CASE)} In the base case we cope with $i = 1$. 
Let $(u,v)$ be a $([3],[2])$-constant pair of triples in $A$ satisfying  
$\orbitA \orbitA \orbitB(u,v)$. We now show that there is $u' \in A^3$ such as $u$ with a difference that $(u'[3],v[3]) \in \orbitA$.
By the fact that $\structB$ has free amalgamation over $\orbitN$ we have that there exists $p, r \in A^3$ satisfying the conditions from  Figure~\ref{fig:cncbasestep}. Since $\structA$ is homogeneous, there exists an automorphism $\alpha$ of $\structA$ such that $\alpha(u[i])= \alpha(p[i])$ and $\alpha(v[i]) = \alpha(r[i])$ for $i \in [3]$. Now $\alpha^{-1}(p'[3])$ is the desired $u'[3]$.
Since $\structB$ contains $\orbitA$ it follows that 
$(D_1(u'),D_1(v)) \in \orbitA$. By~(\ref{eq:D1}), we have $D_1(u') = D_1(u)$, and hence 
$(D_1(u),D_1(v)) \in \orbitA$. Since $\orbitA \subseteq \subsetA$,it completes the proof of the base case for the first item.


\begin{figure}
\tikz {
\node (u) at (0,4) {$p[1] = p[2] = p[3]$};
\node (v12) at (3,4) {$r[1] = r[2]$};
\node (v3) at (2,2) {$r[3]$};
\node (u3prim) at (0,1) {$p'[3]$};

\draw (u) -- node[above] {$\orbitA$} (v12);
\draw (u) -- node[above] {$\orbitB$} (v3);
\draw (u) -- node[left] {$\orbitO$} (u3prim);
\draw (v12) -- node[left] {$\orbitO$} (v3);
\draw (u3prim) -- node[above] {$\orbitA$} (v3);
\path (u3prim) edge [out=290, in=270] node[below] {$\orbitP$} (v12)
}
\centering
 \caption{A structure in $\Age(\structA)$ needed in the proof of the base case of Claim~\ref{claim:bnbnoJonsson}. If both $\orbitA$ and $\orbitB$ are anti reflexive, then $\orbitO$ and $\orbitP$ are $\orbitN$. If $\orbitA$ is $=$, then $\orbitO$ and $\orbitP$ are $\orbitB$. If $\orbitB$ is $=$, then $\orbitO$ is $\orbitA$ and $\orbitP$ is $\orbitN$.}
\label{fig:cncbasestep}
\end{figure}

\begin{figure}
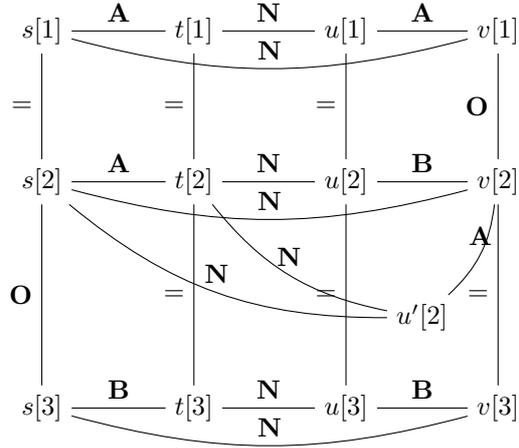

\tikz {
\node (s1) at (0,5) {$s[1]$};
\node (s2) at (0,3) {$s[2]$};
\node (s3) at (0,0) {$s[3]$};
\node (t1) at (2,5) {$t[1]$};
\node (t2) at (2,3) {$t[2]$};
\node (t3) at (2,0) {$t[3]$};
\node (u1) at (4,5) {$u[1]$};
\node (u2) at (4,3) {$u[2]$};
\node (u3) at (4,0) {$u[3]$};
\node (v1) at (6,5) {$v[1]$};
\node (v2) at (6,3) {$v[2]$};
\node (v3) at (6,0) {$v[3]$};
\node (u2prim) at (5,1.2) {$u'[2]$};

\draw (s1) -- node[above] {$\orbitA$} (t1);
\draw (s2) -- node[above] {$\orbitA$} (t2);
\draw (s3) -- node[above] {$\orbitB$} (t3);

\draw (t1) -- node[above] {$\orbitN$} (u1);
\draw (t2) -- node[above] {$\orbitN$} (u2);
\draw (t3) -- node[above] {$\orbitN$} (u3);

\draw (u1) -- node[above] {$\orbitA$} (v1);
\draw (u2) -- node[above] {$\orbitB$} (v2);
\draw (u3) -- node[above] {$\orbitB$} (v3);

\draw (s1) -- node[left] {$=$} (s2);
\draw (s2) -- node[left] {$\orbitO$} (s3);

\draw (t1) -- node[left] {$=$} (t2);
\draw (t2) -- node[left] {$=$} (t3);

\draw (u1) -- node[left] {$=$} (u2);
\draw (u2) -- node[left] {$=$} (u3);

\draw (v1) -- node[left] {$\orbitO$} (v2);
\draw (v2) -- node[left] {$=$} (v3);

\path (u2prim) edge[bend right=20] node[above] {$\orbitA$} (v2);
\path (u2prim) edge[bend left=20] node[above] {$\orbitN$} (s2);
\path (u2prim) edge[bend left=20] node[above] {$\orbitN$} (t2);

\path (s1) edge[bend right=15] node[above] {$\orbitN$} (v1);
\path (s2) edge[bend right=15] node[above] {$\orbitN$} (v2);
\path (s3) edge[bend right=15] node[above] {$\orbitN$} (v3)
}
\centering
\caption{Since the age of $\structA$ has free amalgamation over $\orbitN$,  it is straightforward to show that there are four vectors $s,t,u,v \in A^3$ and a single additional element $u'[2]$ described by the diagram above. By homogeneity of $\structA$ we may assume that  $u,v \in A^3$ are the vectors from the proof of  Claim~\ref{claim:bnbnoJonsson}. For edges not depicted in the picture the label either follows by the depicted ones and the semi-transitivity of equality or it is $\orbitN$.}
\label{fig:bnbextension}
\end{figure}


For the second bullet, let $(u,v)$ be a $([3],\{2,3\})$-constant pair of triples 
satisfying $\orbitA \orbitB \orbitB (u,v)$. 
Consult now Figure~\ref{fig:bnbextension} to see that there exists 
a  $([2],[3])$-constant pair of triples $s,t \in A^3$ satisfying  
$\orbitA \orbitA \orbitC (s,t)$
such that 
\begin{itemize}
\item $a_1 =(s[1], t[1], u[1], v[1])$ is a $(\orbitA, \orbitN, \orbitN,\orbitN, \orbitN, \orbitA)$-tuple, 
\item $a_2 = (s[2], t[2], u[2], v[2])$ is a $(\orbitA, \orbitN, \orbitN,\orbitN, \orbitN, \orbitB)$-tuple, and
\item $a_3 = (s[3], t[3], u[3], v[3])$ is a $(\orbitB, \orbitN, \orbitN,\orbitN, \orbitN, \orbitB)$-tuple.
\end{itemize}
Let now $u'$ be as $u$ except for the second coordinate so that the tuple
$a'_2 = (s[2], t[2], u'[2], v[2])$ is a $(\orbitA, \orbitN, \orbitN,\orbitN, \orbitN, \orbitA)$-tuple. 
Now all $a_1, a'_2, a_3$ are tuples in $R$.
Since $(D_1(s), D_1(t)) \in \subsetA$ and $\subsetA + R = \subsetA$ we have that 
$(D_1(u'), D_1(v)) \in \subsetA$. By~(\ref{eq:Di}) we have $D_1(u') = D_1(u)$, and hence
$(D_1(u), D_1(v)) \in \subsetA$.
It completes the base case.  

\smallskip
\noindent
\textbf{(INDUCTION STEP)}
For the induction step, consider any $([3],[2])$-constant pair of triples $(u,v)$
in $A$
satisfying $\orbitA \orbitA \orbitB (u,v)$. We want to show that
$(D_{i+1}(u), D_{i+1}(v)) \in \subsetA$.
To that end consider  $v'$ such that $v'[1] =  v[1]$ and $v'[2] = v'[3] = v[3]$.
Observe that $(u,v')$ is a $([3],\{ 2,3 \})$-constant pair of triples in $A$ satisfying $\orbitA \orbitB \orbitB (u,v')$. By the induction hypothesis, it follows that ($D_i(u),D_i(v')) \in \subsetA$.
By~(\ref{eq:Dii+1}), we have $D_{i}(v) = D_{i+1}(v')$
and clearly $D_{i+1}(u) = D_i(u)$. It follows that 
$(D_{i+1}(u), D_{i+1}(v)) \in \subsetA$.
Now, as in the base case we show that
$(D_{i+1}(u), D_{i+1}(v)) \in \subsetA$ for all $([3], \{ 2,3 \})$-constant 
pair of triples $(u,v)$ in $\orbitA$ satisfying $\orbitA \orbitB \orbitB(u,v)$. It completes the proof of the claim.  
\end{proof}

Finally we have to show that $(D_n(u),D_n(v)) \in \subsetA$ for some $([3], [3])$-constant pair of triples $u,v$ satisfying $\orbitB \orbitB \orbitB (u,v)$.
Since the age of $\structA$ has free amalgamation and the structure is homogeneous, there exists a $(\{ 2, 3\}, [3])$-constant pair of triples $(s,t)$ in $A$ satisfying  
$\orbitA \orbitB \orbitB (s,t)$
such that 
\begin{itemize}
\item $a_1 =(s[1], t[1], u[1], v[1])$ is a $(\orbitA, \orbitN, \orbitN,\orbitN, \orbitN, \orbitB)$-tuple, 
\item $a_2 = (s[2], t[2], u[2], v[2])$ is a $(\orbitB, \orbitN, \orbitN,\orbitN, \orbitN, \orbitB)$-tuple, and
\item $a_3 = (s[3], t[3], u[3], v[3])$ is a $(\orbitB, \orbitN, \orbitN,\orbitN, \orbitN, \orbitB)$-tuple
\end{itemize}
(See~Figure~\ref{fig:bnbfinalstep}.)
Let now $u'$ be as $u$ except for  $u'[1]$ which is such that 
$a'_1 = (s[1], t[1], u'[1], v[1])$ is a $(\orbitA, \orbitN, \orbitN,\orbitN, \orbitN, \orbitA)$-tuple. Now all $a'_1, a_2, a_3$ are tuples in $R$.
Since $(D_n(s), D_n(t)) \in \subsetA$ and $\subsetA + R = \subsetA$ we have that 
$(D_n(u'), D_n(v)) \in \subsetA$. By~(\ref{eq:Dn}), we have $D_n(u') = D_n(u)$. It follows that
$(D_n(u), D_n(v)) \in \subsetA$ for $(u,v)$ satisfying $\orbitB \orbitB \orbitB (u,v)$. Since $\orbitB \subsetneq \subsetA$, we arrived at a contradiction. It completes the proof of the lemma. 
\end{proof}

\begin{figure}
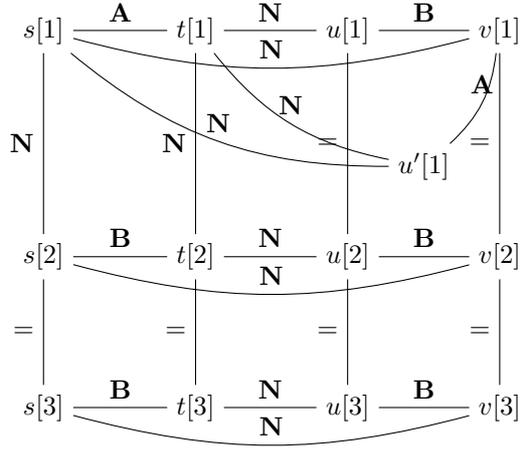

\tikz {
\node (s1) at (0,6) {$s[1]$};
\node (s2) at (0,3) {$s[2]$};
\node (s3) at (0,1) {$s[3]$};
\node (t1) at (2,6) {$t[1]$};
\node (t2) at (2,3) {$t[2]$};
\node (t3) at (2,1) {$t[3]$};
\node (u1) at (4,6) {$u[1]$};
\node (u2) at (4,3) {$u[2]$};
\node (u3) at (4,1) {$u[3]$};
\node (v1) at (6,6) {$v[1]$};
\node (v2) at (6,3) {$v[2]$};
\node (v3) at (6,1) {$v[3]$};
\node (u2prim) at (5,4.2) {$u'[1]$};

\draw (s1) -- node[above] {$\orbitA$} (t1);
\draw (s2) -- node[above] {$\orbitB$} (t2);
\draw (s3) -- node[above] {$\orbitB$} (t3);

\draw (t1) -- node[above] {$\orbitN$} (u1);
\draw (t2) -- node[above] {$\orbitN$} (u2);
\draw (t3) -- node[above] {$\orbitN$} (u3);

\draw (u1) -- node[above] {$\orbitB$} (v1);
\draw (u2) -- node[above] {$\orbitB$} (v2);
\draw (u3) -- node[above] {$\orbitB$} (v3);

\draw (s1) -- node[left] {$\orbitN$} (s2);
\draw (s2) -- node[left] {$=$} (s3);

\draw (t1) -- node[left] {$\orbitN$} (t2);
\draw (t2) -- node[left] {$=$} (t3);

\draw (u1) -- node[left] {$=$} (u2);
\draw (u2) -- node[left] {$=$} (u3);

\draw (v1) -- node[left] {$=$} (v2);
\draw (v2) -- node[left] {$=$} (v3);

\path (u2prim) edge[bend right=20] node[above] {$\orbitA$} (v1);
\path (u2prim) edge[bend left=20] node[above] {$\orbitN$} (s1);
\path (u2prim) edge[bend left=20] node[above] {$\orbitN$} (t1);

\path (s1) edge[bend right=15] node[above] {$\orbitN$} (v1);
\path (s2) edge[bend right=15] node[above] {$\orbitN$} (v2);
\path (s3) edge[bend right=15] node[above] {$\orbitN$} (v3)
}
\centering
\caption{Triples $s,t,u,v$ and $u'$ used for a final step in the proof of Lemma~\ref{lem:bnbnoJonsson}. If it does not follows from the labels in the diagram and the semi-transitivity of equality, all omitted edges are labelled with $\orbitN$.} 
\label{fig:bnbfinalstep}
\end{figure}

The following lemma is a version of Lemma~\ref{lem:bnbnoJonsson} but with a slightly different condition on the second tuple and a bit more complicated proof. 


\begin{lemma}
\label{lem:beqbnoJonsson}
Let $\structB$ be a  first-order expansion of a finitely bounded homogeneous symmetric binary core $\structA$ whose age has free amalgamation. If $\structB$ pp-defines a self-complementary implication $R: \subsetA \rightarrow \subsetA$ containing  
\begin{itemize}
\item a $(\orbitA, \orbitN, \orbitN, \orbitN, \orbitN, \orbitA)$-tuple and 
\item  a $(\orbitB, \orbitB, =, = \orbitB,  \orbitB)$-tuple, 
\end{itemize}
for some orbitals $\orbitA \subseteq \subsetA$ and $\orbitB \nsubseteq \subsetA$, then $\structB$ is not preserved by any chain of quasi directed J\'{o}nsson operations. 
\end{lemma}

\begin{proof}
The proof goes along the lines of the proof of Lemma~\ref{lem:bnbnoJonsson}. We assume on the contrary that there exists a chain $(D_1, \ldots,D_n)$ of quasi directed J\'{o}nsson operations preserving $R$ and show that it contradicts the fact that $\structB$ contains all orbitals.


In the formulation of the claim below we deal with   $([3],[2])$-constant pairs of triples $(u,v)$ in $A$
satisfying $\orbitA\orbitA\orbitB(u,v)$ as well as  
$([3],\{2,3\})$-constant  pairs of triples $(u,v)$ in $A$ 
satisfying $\orbitA\orbitB\orbitB(u,v)$ of which we already know they exist --- see the discussion above Claim~\ref{claim:bnbnoJonsson}.
For other $(U,V)$-constant tuples in the claim below consult Figures~\ref{fig:beqb22extension} and~\ref{fig:beqb2323extension}. Indeed, $(u,v)$ in the former figure is a 
$([2], [2])$-constant pair of triples 
satisfying $\orbitA\orbitA\orbitB(u,v)$ while
$(u,v)$ in the latter is a $(\{2,3\},$
$ \{2,3\})$-constant pair of triples in $A$ satisfying  $\orbitA\orbitB\orbitB(s,t)$.

\begin{claim}
\label{claim:beqbnoJonsson}
For all $i \in [n]$ we have both of the following.
\begin{itemize}
\item  For all $([3],[2])$-constant and all $([2], [2])$-constant pairs of triples $(u,v)$ in $A$
satisfying $\orbitA\orbitA\orbitB(u,v)$ it is the truth that
$(D_i(u), D_i(v)) \in \orbitA$. 
\item For all $([3],\{2,3\})$-constant and all $(\{2,3\}, \{2,3\})$-constant  pairs of triples $(u,v)$ in $A$ 
satisfying $\orbitA\orbitB\orbitB(u,v)$ it is the truth that
$(D_i(u), D_i (v)) \in \orbitA$.  
\end{itemize}
\end{claim}

\begin{figure}
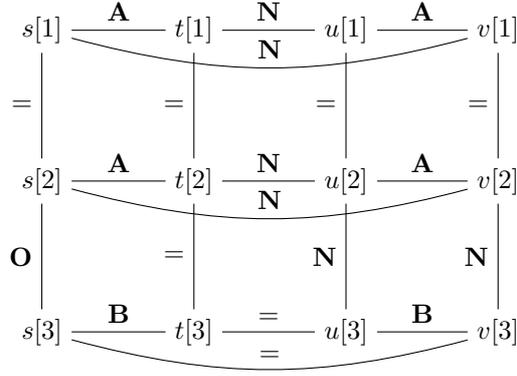

\tikz {
\node (s1) at (0,5) {$s[1]$};
\node (s2) at (0,3) {$s[2]$};
\node (s3) at (0,1) {$s[3]$};
\node (t1) at (2,5) {$t[1]$};
\node (t2) at (2,3) {$t[2]$};
\node (t3) at (2,1) {$t[3]$};
\node (u1) at (4,5) {$u[1]$};
\node (u2) at (4,3) {$u[2]$};
\node (u3) at (4,1) {$u[3]$};
\node (v1) at (6,5) {$v[1]$};
\node (v2) at (6,3) {$v[2]$};
\node (v3) at (6,1) {$v[3]$};

\draw (s1) -- node[above] {$\orbitA$} (t1);
\draw (s2) -- node[above] {$\orbitA$} (t2);
\draw (s3) -- node[above] {$\orbitB$} (t3);

\draw (t1) -- node[above] {$\orbitN$} (u1);
\draw (t2) -- node[above] {$\orbitN$} (u2);
\draw (t3) -- node[above] {$=$} (u3);

\draw (u1) -- node[above] {$\orbitA$} (v1);
\draw (u2) -- node[above] {$\orbitA$} (v2);
\draw (u3) -- node[above] {$\orbitB$} (v3);

\draw (s1) -- node[left] {$=$} (s2);
\draw (s2) -- node[left] {$\orbitO$} (s3);

\draw (t1) -- node[left] {$=$} (t2);
\draw (t2) -- node[left] {$=$} (t3);

\draw (u1) -- node[left] {$=$} (u2);
\draw (u2) -- node[left] {$\orbitN$} (u3);

\draw (v1) -- node[left] {$=$} (v2);
\draw (v2) -- node[left] {$\orbitN$} (v3);


\path (s1) edge[bend right=15] node[above] {$\orbitN$} (v1);
\path (s2) edge[bend right=15] node[above] {$\orbitN$} (v2);
\path (s3) edge[bend right=15] node[above] {$=$} (v3)
}
\centering
\caption{A substructure of $\structA$ needed for the proof of the base case of induction in the proof of Claim~\ref{claim:beqbnoJonsson}.} 
\label{fig:beqb22extension}
\end{figure}


\begin{proof}
We prove the claim by the induction on $i \in [n]$.

\textbf{(BASE CASE)} 
 The fact that for all  $([3],[2])$-constant pairs of triples $(u,v)$
 in $A$ satisfying  
$\orbitA \orbitA \orbitB(u,v) $ we have $(D_1(u),D_1(v)) \in \orbitA$ has been already proved while proving Claim~\ref{claim:bnbnoJonsson}.

Let now  $(u,v)$ be any $([2],[2])$-constant pair of triples satisfying $\orbitA \orbitA \orbitB(u,v)$. 
Since the age of $\structA$ has free amalgamation (over $\orbitN$) and the structure is homogeneous,  there exists a $([2],[3])$-constant pair of triples $s,t$ in $A$ (as in Figure~\ref{fig:beqb22extension})
such that
\begin{itemize}
\item $a_1 =(s[1], t[1], u[1], v[1])$ is a $(\orbitA, \orbitN, \orbitN,\orbitN, \orbitN, \orbitA)$-tuple, 
\item $a_2 = (s[2], t[2], u[2], v[2])$ is a $(\orbitA, \orbitN, \orbitN,\orbitN, \orbitN, \orbitA)$-tuple, and
\item $a_3 = (s[3], t[3], u[3], v[3])$ is a $(\orbitB, \orbitB, =,=, \orbitB, \orbitB)$-tuple.
\end{itemize}
Hence all $a_1, a_2, a_3 \in R$. Since $(D_1(s), D_1(t))\in \subsetA$ and $\subsetA + R = \subsetA$ we have $(D_1(u), D_1(v)) \in \orbitA$. It completes the proof of the base case for the first bullet.

\begin{figure}
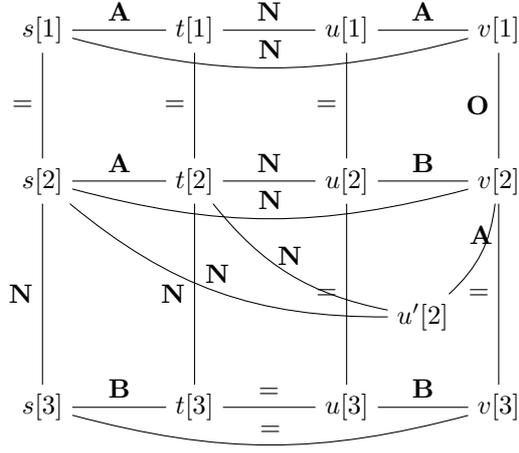

\tikz {
\node (s1) at (0,5) {$s[1]$};
\node (s2) at (0,3) {$s[2]$};
\node (s3) at (0,0) {$s[3]$};
\node (t1) at (2,5) {$t[1]$};
\node (t2) at (2,3) {$t[2]$};
\node (t3) at (2,0) {$t[3]$};
\node (u1) at (4,5) {$u[1]$};
\node (u2) at (4,3) {$u[2]$};
\node (u3) at (4,0) {$u[3]$};
\node (v1) at (6,5) {$v[1]$};
\node (v2) at (6,3) {$v[2]$};
\node (v3) at (6,0) {$v[3]$};
\node (u2prim) at (5,1.2) {$u'[2]$};

\draw (s1) -- node[above] {$\orbitA$} (t1);
\draw (s2) -- node[above] {$\orbitA$} (t2);
\draw (s3) -- node[above] {$\orbitB$} (t3);

\draw (t1) -- node[above] {$\orbitN$} (u1);
\draw (t2) -- node[above] {$\orbitN$} (u2);
\draw (t3) -- node[above] {$=$} (u3);

\draw (u1) -- node[above] {$\orbitA$} (v1);
\draw (u2) -- node[above] {$\orbitB$} (v2);
\draw (u3) -- node[above] {$\orbitB$} (v3);

\draw (s1) -- node[left] {$=$} (s2);
\draw (s2) -- node[left] {$\orbitN$} (s3);

\draw (t1) -- node[left] {$=$} (t2);
\draw (t2) -- node[left] {$\orbitN$} (t3);

\draw (u1) -- node[left] {$=$} (u2);
\draw (u2) -- node[left] {$=$} (u3);

\draw (v1) -- node[left] {$\orbitO$} (v2);
\draw (v2) -- node[left] {$=$} (v3);

\path (u2prim) edge[bend right=20] node[above] {$\orbitA$} (v2);
\path (u2prim) edge[bend left=20] node[above] {$\orbitN$} (s2);
\path (u2prim) edge[bend left=20] node[above] {$\orbitN$} (t2);

\path (s1) edge[bend right=15] node[above] {$\orbitN$} (v1);
\path (s2) edge[bend right=15] node[above] {$\orbitN$} (v2);
\path (s3) edge[bend right=15] node[above] {$=$} (v3)
}
\centering
\caption{Triples $u,v,s,t$ and $u'$ which play a role in the proof of
the  the base case of induction in Claim~\ref{claim:beqbnoJonsson}.}
\label{fig:beqb323extension}
\end{figure}

For the second bullet, let $(u,v)$ be now a $([3],\{2,3\})$-constant pair
of triples in $A$ satisfying $\orbitA \orbitB \orbitB(u,v)$.
Again, since the age of $\structA$ has free amalgamation and the structure is homogeneous,  there exists a $([2],[2])$-constant pair of triples $(s,t)$ in $A$ (as in Figure~\ref{fig:beqb323extension}) 
such that 
\begin{itemize}
\item $a_1 =(s[1], t[1], u[1], v[1])$ is a $(\orbitA, \orbitN, \orbitN,\orbitN, \orbitN, \orbitA)$-tuple, 
\item $a_2 = (s[2], t[2], u[2], v[2])$ is a $(\orbitA, \orbitN, \orbitN,\orbitN, \orbitN, \orbitB)$-tuple, and
\item $a_3 = (s[3], t[3], u[3], v[3])$ is a $(\orbitB, \orbitB, =,=, \orbitB, \orbitB)$-tuple.
\end{itemize}

Let now $u'$ be as $u$ except for the second coordinate. We choose $u'[2]$ so that 
$a'_2 = (s[2], t[2], u'[2], v[2])$ is a $(\orbitA, \orbitN, \orbitN,\orbitN, \orbitN, \orbitA)$-tuple.  Now all $a_1, a'_2, a_3$ are tuples in $R$.
Since $(D_1(s), D_1(t)) \in \subsetA$ and $\subsetA + R = \subsetA$ we have that 
$(D_1(u'), D_1(v)) \in \subsetA$. By~(\ref{eq:Di}), we have $D_1(u') = D_1(u)$, and hence
$(D_1(u), D_1(v)) \in \subsetA$.

Finally, let $(u,v)$ be a $(\{2,3\},\{2,3\})$-constant pair of triples in $A$ satisfying $\orbitA\orbitB\orbitB(u,v)$. Since the age of $\structA$ has free amalgamation and the structure is homogeneous, there is some $(\{ 2, 3 \}, [3])$-constant pair of triples $(s, t)$ in $A$ satisfying $\orbitA\orbitB\orbitB(s,t)$ (as in Figure~\ref{fig:beqb2323extension})
such that 
\begin{itemize}
\item $a_1 =(s[1], t[1], u[1], v[1])$ is a $(\orbitA, \orbitN, \orbitN,\orbitN, \orbitN, \orbitA)$-tuple, 
\item $a_2 = (s[2], t[2], u[2], v[2])$ is a $(\orbitB, \orbitB, =,=, 
\orbitB, \orbitB)$-tuple, and
\item $a_3 = (s[3], t[3], u[3], v[3])$ is a $(\orbitB, \orbitB, =,=, 
\orbitB, \orbitB)$-tuple.
\end{itemize}

Observe that all $a_1,a_2, a_3 \in R$.
Since $(D_1(s), D_1(t) \in \subsetA)$ and $\subsetA + R = \subsetA$, it follows that $(D_1(u), D_1(v) \in \subsetA$. It completes the proof of the base case.

\begin{figure}
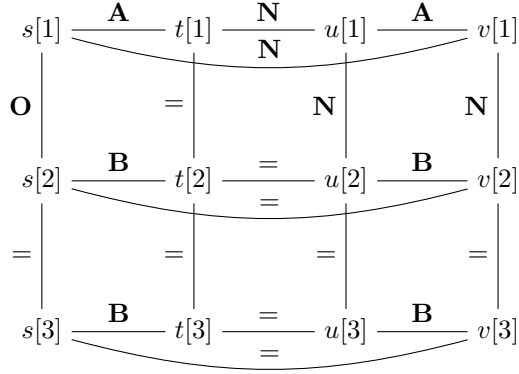

\tikz {
\node (s1) at (0,5) {$s[1]$};
\node (s2) at (0,3) {$s[2]$};
\node (s3) at (0,1) {$s[3]$};
\node (t1) at (2,5) {$t[1]$};
\node (t2) at (2,3) {$t[2]$};
\node (t3) at (2,1) {$t[3]$};
\node (u1) at (4,5) {$u[1]$};
\node (u2) at (4,3) {$u[2]$};
\node (u3) at (4,1) {$u[3]$};
\node (v1) at (6,5) {$v[1]$};
\node (v2) at (6,3) {$v[2]$};
\node (v3) at (6,1) {$v[3]$};

\draw (s1) -- node[above] {$\orbitA$} (t1);
\draw (s2) -- node[above] {$\orbitB$} (t2);
\draw (s3) -- node[above] {$\orbitB$} (t3);

\draw (t1) -- node[above] {$\orbitN$} (u1);
\draw (t2) -- node[above] {$=$} (u2);
\draw (t3) -- node[above] {$=$} (u3);

\draw (u1) -- node[above] {$\orbitA$} (v1);
\draw (u2) -- node[above] {$\orbitB$} (v2);
\draw (u3) -- node[above] {$\orbitB$} (v3);

\draw (s1) -- node[left] {$\orbitO$} (s2);
\draw (s2) -- node[left] {$=$} (s3);

\draw (t1) -- node[left] {$=$} (t2);
\draw (t2) -- node[left] {$=$} (t3);

\draw (u1) -- node[left] {$\orbitN$} (u2);
\draw (u2) -- node[left] {$=$} (u3);

\draw (v1) -- node[left] {$\orbitN$} (v2);
\draw (v2) -- node[left] {$=$} (v3);

\path (s1) edge[bend right=15] node[above] {$\orbitN$} (v1);
\path (s2) edge[bend right=15] node[above] {$=$} (v2);
\path (s3) edge[bend right=15] node[above] {$=$} (v3)
}
\centering
\caption{Triples $u,v,s,t$ which play a role in the proof of
the  base case of induction in the proof of Claim~\ref{claim:beqbnoJonsson}.}
\label{fig:beqb2323extension}
\end{figure}

\smallskip
\noindent
\textbf{(INDUCTION STEP)}
The fact that 
$(D_{i+1}(u), D_{i+1}(v)) \in \orbitA$.
for all $([3],[2])$-constant pair of triples $(u,v)$ in $A$
satisfying $\orbitA \orbitA \orbitB (u,v)$ is shown in the same way 
as in the proof of Lemma~\ref{lem:bnbnoJonsson}.

The three remaining facts we need to go through the induction step may be proved in the exactly same way as for the base case but one needs to replace $D_1$ with $D_{i+1}$ in all  reasonings.  
\end{proof}

The final step of the proof is to show that $(D_n(u),D_n(v)) \in \subsetA$ for a pair of $([3], [3])$-constant pair of triples $(u,v)$ in $A$ satisfying $\orbitB \orbitB \orbitB (u,v)$.
To this end, consider such $u,v \in A^3$.
Since $\structA$ is homogeneous and its age has free amalgamation over $\orbitN$, there exists a $(\{ 2, 3\}, \{2,3\})$-constant pair of triples $(s,t)$ in $A$ (consult Figure~\ref{fig:beqbfinalstep}) satisfying  
$\orbitA \orbitB \orbitB (s,t)$
such that 
\begin{itemize}
\item $a_1 =(s[1], t[1], u[1], v[1])$ is a $(\orbitA, \orbitN, \orbitN,\orbitN, \orbitN, \orbitB)$-tuple, 
\item $a_2 = (s[2], t[2], u[2], v[2])$ is a $(\orbitB,\orbitB, =, =, \orbitB, \orbitB)$-tuple, and
\item $a_3 = (s[3], t[3], u[3], v[3])$ is a $(\orbitB,\orbitB, =, =, \orbitB, \orbitB)$-tuple.
\end{itemize}
Let now $u'$ be as $u$ except for the first coordinate. We choose $u'[1]$ so that 
$a'_1 = (s[1], t[1], u'[1], v[1])$ is a $(\orbitA, \orbitN, \orbitN,\orbitN, \orbitN, \orbitA)$-tuple.  Now all $a'_1, a_2, a_3$ are tuples in $R$.
Since $(D_n(s), D_n(t)) \in \orbitA$ and $\subsetA + R = \subsetA$ we have that 
$(D_n(u'), D_n(v)) \in \subsetA$. By~(\ref{eq:Dn}), we have $D_n(u') = D_n(u)$ and hence
$(D_n(u), D_n(v)) \in \subsetA$ for $u,v$ satisfying $\orbitB \orbitB \orbitB (u,v)$. Since $D_n$ preserves $\orbitB$, we arrived at a contradiction. It completes the proof of the lemma. 
\end{proof}

\begin{figure}
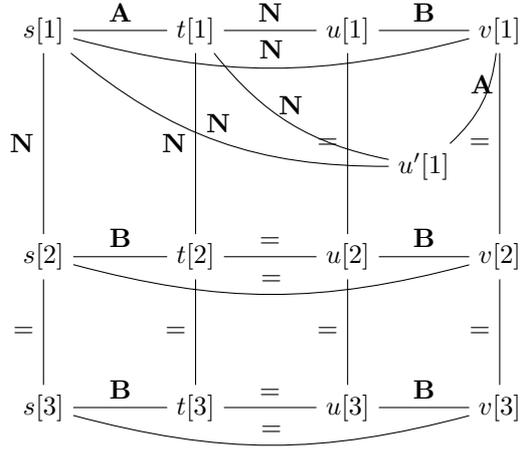

\tikz {
\node (s1) at (0,6) {$s[1]$};
\node (s2) at (0,3) {$s[2]$};
\node (s3) at (0,1) {$s[3]$};
\node (t1) at (2,6) {$t[1]$};
\node (t2) at (2,3) {$t[2]$};
\node (t3) at (2,1) {$t[3]$};
\node (u1) at (4,6) {$u[1]$};
\node (u2) at (4,3) {$u[2]$};
\node (u3) at (4,1) {$u[3]$};
\node (v1) at (6,6) {$v[1]$};
\node (v2) at (6,3) {$v[2]$};
\node (v3) at (6,1) {$v[3]$};
\node (u2prim) at (5,4.2) {$u'[1]$};

\draw (s1) -- node[above] {$\orbitA$} (t1);
\draw (s2) -- node[above] {$\orbitB$} (t2);
\draw (s3) -- node[above] {$\orbitB$} (t3);

\draw (t1) -- node[above] {$\orbitN$} (u1);
\draw (t2) -- node[above] {$=$} (u2);
\draw (t3) -- node[above] {$=$} (u3);

\draw (u1) -- node[above] {$\orbitB$} (v1);
\draw (u2) -- node[above] {$\orbitB$} (v2);
\draw (u3) -- node[above] {$\orbitB$} (v3);

\draw (s1) -- node[left] {$\orbitN$} (s2);
\draw (s2) -- node[left] {$=$} (s3);

\draw (t1) -- node[left] {$\orbitN$} (t2);
\draw (t2) -- node[left] {$=$} (t3);

\draw (u1) -- node[left] {$=$} (u2);
\draw (u2) -- node[left] {$=$} (u3);

\draw (v1) -- node[left] {$=$} (v2);
\draw (v2) -- node[left] {$=$} (v3);

\path (u2prim) edge[bend right=20] node[above] {$\orbitA$} (v1);
\path (u2prim) edge[bend left=20] node[above] {$\orbitN$} (s1);
\path (u2prim) edge[bend left=20] node[above] {$\orbitN$} (t1);

\path (s1) edge[bend right=15] node[above] {$\orbitN$} (v1);
\path (s2) edge[bend right=15] node[above] {$=$} (v2);
\path (s3) edge[bend right=15] node[above] {$=$} (v3)
}
\centering
\caption{The figure depicts triples $s,t,u,v$ and $u'$ used in the final step of the proof of Lemma~\ref{lem:beqbnoJonsson}.}
\label{fig:beqbfinalstep}
\end{figure}

\noindent
We will now conclude Section~\ref{sect:nondegen} by a theorem which follows directly from Proposition~\ref{prop:nondegen_reduce} and Lemmas~\ref{lem:bnbnoJonsson} and~\ref{lem:beqbnoJonsson}.

\begin{theorem}
\label{thm:nondegen}
Let $\structB$ be a first-order expansion of a finitely bounded homogeneous symmetric binary core which pp-defines $R_1:\subsetA \rightarrow \subsetB, R_2: \subsetB \rightarrow \subsetA$  such that $\bipartite_{R_1, R_2}$ has a non-degenerated component. Then $\structB$ is not preserved by any chains of quasi directed J\'{o}nnson operations.  
\end{theorem}


\subsection{All non-trivial components are degenerated}
\label{sect:alldegen}

By the results of the previous subsection we may assume that all  nontrivial components in $\bipartite_{R_1, R_2}$ are degenerated. 
We reduce this situation via Proposition~\ref{prop:degenreduce} to three simply cases.
To this end we need a number of definitions and observations. 

Besides the $\circ$-composition defined earlier, we also define the $\bowtie$-composition of two relations.  
\begin{definition}
\label{def:bowtie}
    Let $R_1, R_2$ be two quaternary relations, then $R_1 \bowtie R_2$ is
    $$R_3(x_1, x_2, x_3, x_4) \equiv \exists y_1 y_2~R_3(x_1, x_2, y_1, y_2) \wedge R_2(y_2, y_1,x_1, x_2).$$
\end{definition} 

The following observation may be proved in the exactly same way as Observation~\ref{obs:selfcomplementary}

\begin{observation}
\label{obs:selfcomp_bowtie}
Let $R_1: \mathcal{A} \rightarrow \mathcal{B}$ and $R_2: \mathcal{B} \rightarrow \mathcal{A}$ be complementary implications. Then $R_1 \bowtie  R_2$ is a self-complementary  $(\subsetA \rightarrow \subsetA)$-implication. 
\end{observation}

We will write $(R_1 \bowtie R_2)^n$ as a shorthand for  the expression $(( \cdots (((R_1 \bowtie R_2) \bowtie R_1) \bowtie R_2) \bowtie \cdots \bowtie R_1) \circ R_2)$ where both $R_1$ and $R_2$ occur $n$ times. 
The next observation is analogous to Observation~\ref{obs:circbipartite}.

\begin{observation}
\label{obs:bowtiebipartite}
  Let $R_1: \subsetA \rightarrow \subsetB$ and $R_2: \subsetB \rightarrow \subsetA$ be complementary relations. Then, for all $n \geq 1$ we have both of the following: 
  \begin{itemize}
      \item $S_{2n} \equiv (R_1 \bowtie R_2)^n$ contains an $(\orbitO,\ldots, \orbitP)$-tuple
      iff there is a path in $\bipartite_{R_1, R_2}$ of length $2n$ from $\orbitO_L$ to $\orbitP_L$;
      \item $S_{2n+1} \equiv (R_1 \bowtie R_2)^n \circ R_1$ contains an $(\orbitO, \ldots, \orbitP)$-tuple
      iff there is a path in $\bipartite_{R_1, R_2}$ of length $2n+1$ from $\orbitO_L$ to $\orbitP_R$. 
  \end{itemize}
\end{observation}

\noindent
Finally, the next observation resembles  Observation~\ref{obs:tuplescomposition} with a proof following strictly from a definition of the $\bowtie$-composition.

\begin{observation}
\label{obs:bowtietuples}
Let $R_1 : \subsetA \rightarrow \subsetB$ and $R_2: \subsetB \rightarrow \subsetC$ be such that $R_1$ contains a $(\orbitA, \ldots, \orbitB)$-tuple $t_1$ and $R_2$ contains a $(\orbitB,\ldots, \orbitC)$-tuple $t_2$. Then $R_3 := R_1 \bowtie R_2$ is a 
$(\subsetA \rightarrow \subsetC)$-implication containing a $(\orbitA, \ldots, \orbitC)$-tuple $t_3$ which is
  essentially ternary if both $t_1, t_2$ are essentilly ternary.
\end{observation}

\noindent
For a vertex $\orbitO_P$ with $P \in \{ L, R \}$ and $S \in \{ L, R \}$ we define:
\begin{align}
    \text{FReach}_S(\orbitO_P)  = \{ \orbitP_S \mid \text {there is a path in } \bipartite_{R_1, R_2} \text{ from  } \orbitO_P \text{ to } \orbitP_S \}\nonumber\\
       \text{BReach}_S(\orbitO_P)  = \{ \orbitP_S \mid \text {there is a path in } \bipartite_{R_1, R_2} \text{ from  } \orbitO_P \text{ to } \orbitP_S  \}\nonumber
\end{align}

\begin{observation}
\label{obs:defnreach}
Let $R_1, R_2$ be a pair of complementary relations pp-definable in a first-order expansion 
of a finitely bounded homogeneous symmetric binary core  such that all non-trivial components in $\bipartite_{R_1, R_2}$ are degenerated and $\orbitO$ an orbital such that $\bipartite_{R_1, R_2}$ contains a non-trivial $\orbitO$-degenerated component, then both 
$\text{FReach}_P(\orbitO_P)$ and $\text{BReach}_P(\orbitO_P)$
for $P \in \{ L, R \}$ are pp-definable in $\structB$.
\end{observation}

\begin{proof}
We first show that
$$\text{FReach}_L(\orbitO_L)(x_1, x_2) \equiv 
\exists y_1 y_2~\orbitO(y_1, y_2) \wedge R_3(y_1, y_2, x_1, x_2)$$ where 
$$R_3 \equiv (R_1 \bowtie R_2)^n$$
with $n$ being the number of all orbitals in $\structA$.
By Observation~\ref{obs:bowtiebipartite}, a binary relation defined by a formula above contains all orbitals $\orbitP$ such that there is a path from $\orbitO_L$ to $\orbitP_L$ in $\bipartite_{R_1, R_2}$ of length $2n$. It is easy to see that if there is any path from $\orbitO_L$ to $\orbitP_L$, then it is of  length $2n$ or of even length smaller than $2n$. Since there is a path of length two from $\orbitO_L$ to $\orbitO_R$ and back, the formula above defines $\text{FReach}_L(\orbitO_L)$. A proof for $\text{FReach}_R(\orbitO_R)$ is the same we just swap $R_1$ with $R_2$.

For  $\text{BReach}_P(\orbitO_P)$ we just need to invert the direction of arrows. Therefore we first set $R'_i(x_1, x_2, x_3, x_4) \equiv R_i(x_4, x_3, x_2, x_1)$ for $i \in [2]$. 
Now, in order to define $\text{BReach}_L(\orbitO_L)$ we use the above formula where $R_1, R_2$ are replaced with $R'_2, R'_1$. For $\text{BReach}_R(\orbitO_R)$ we replace in the above formula $R_1, R_2$ with $R'_1, R'_2$.
\end{proof}

One of the cases tht come up in Proposition~\ref{prop:degenreduce} is related to the connectedness of a quaternary relation defined as follows.

\begin{definition}
\label{def:connectedness}
We say that a quaternary relation $R$ is \emph{connected} if $\bipartite_{R,R'}$ is connected where $R'(x_1, x_2, x_3, x_4) \equiv R(x_4, x_3, x_2, x_1)$.
\end{definition}

Recall that a quaternary tuple $t$ is \emph{partially-free} if $(t[1], t[4]) \in \orbitN$.
We are now ready for the proposition. 

\begin{proposition}
\label{prop:degenreduce}
Let  $\structB$ be a first-order expansion of a finitely bounded homogeneous symmetric binary core whose age has free amalgamation. If $\structB$ pp-defines a pair of complementary implications $R_1 :\subsetA \rightarrow \subsetB$  and $R_2 :\subsetB \rightarrow \subsetA$ 
such that $\subsetA \neq \subsetB$ and all non-trivial components in $\bipartite_{R_1, R_2}$ are degenerated, then $\structB$
pp-defines one of the following:
\begin{enumerate}
\item \label{degen:nonconnected} a relation $R$ which is not connected and contains at least one non-degenerated tuple,
\item \label{degen:essenternary} a ternary self-complementary $R:\subsetA \rightarrow \subsetA$ such that $\bipartite_{R,R}$ contains a $\orbitA$-degenerated component, a $\orbitB$-degenerated component for some 
$\orbitA \subseteq \subsetA$, $\orbitB \nsubseteq \subsetA$ and a $(\orbitB, \orbitD, \orbitA)$-tuple is in $R$ for some anti-reflexive orbital $\orbitD$,
\item \label{degen:partfree} a self-complementary $R:\subsetA \rightarrow \subsetA$ such that $\bipartite_{R,R}$ contains a $\orbitA$-degenerated component, a $\orbitB$-degenerated component for some anti-reflexive
$\orbitA \subseteq \subsetA$, $\orbitB \nsubseteq \subsetA$ and a 
partially-free tuple is in $R$.
\end{enumerate}
\end{proposition}

\begin{proof}
Since $\subsetA \neq \subsetB$ without loss of generality we can assume that there exists an orbit $\orbitC \subseteq \subsetB \setminus \subsetA$.
Let $\orbitD$-degenerated component for some orbital $\orbitD$ be so that there is a path in $\bipartite_{R_1, R_2}$ from $\orbitC$ to $\orbitD$ but there is no path 
from $\orbitC$ to $\orbitD$ that goes through another degenerated non-trivial component --- we will say that the $\orbitD$-degenerated component is directly above $\orbitC$. In a similar manner, let a $\orbitE$-degenerated component be a one directly below $\orbitC$.

We say that $\pi$ is a direct path from $\orbitE_L$ to $\orbitD_L$ in $\bipartite_{R_1, R_2}$
if it is of the form 
$$ \orbitO^1_{L}, \orbitO^2_{R}, \ldots, \orbitO^{2k-1}_{L}, \orbitO^{2k}_{R}, \orbitO^{2k+1}_L$$
in $\bipartite_{R_1, R_2}$  such that  $\orbitO^1$ is $\orbitE$, $\orbitO^{2k+1}$ is $\orbitD$ and for all $i\in [2k]$ 
we have that $R_1$ contains a $(\orbitO^{i}, \ldots, \orbitO^{i+1})$-tuple $t_i$
if $i$ is odd and that $R_2$ contains a  $(\orbitO^{i}, \ldots, \orbitO^{i+1})$-tuple $t_i$ if $i$ is even. 
We also require  that all $t_i$ perhaps with an exception of $1$ and $2k$ are non-degenerated.
The sentence of tuples $t_i$ associated with $\pi$ will be denoted with $\tau$. 
The essential length of $\pi$ is the number of non-degenerated tuples in $\tau$. It is $2k, 2k-1$ or $2k-2$.

Choose now a direct path $\pi$ from $\orbitE_L$ to $\orbitD_L$ in $\bipartite_{R_1, R_2}$ of the largest essential length. 
Since there is one going through $\orbitC_R$ which is different than both $\orbitE$ and $\orbitD$, the essential length of the chosen $\pi$ is at least 
$2$.

First consider the case where all non-degenerated tuples in $\tau$ are essentially ternary. By Observations~\ref{obs:bowtiebipartite},~\ref{obs:selfcomp_bowtie} and~\ref{obs:bowtietuples}, it follows that $(R_1 \bowtie R_2)^k$ is a self-complementary $(\subsetA \rightarrow \subsetA)$-implication and   contains an essentially ternary  $(\orbitE, \ldots, \orbitD)$-tuple. 
Set a new relation $$R'(x_1, x_2, x_3) \equiv R(x_1, x_2, x_2, x_3) \wedge x_2 = x_3 \wedge \subsetC(x_1, x_2) \wedge \subsetD(x_2,x_3)$$ where
$\subsetC = \text{FReach}_L(\orbitE_L) \cap \text{BReach}_L(\orbitD_L)$ and
$\subsetD = \text{FReach}_R(\orbitE_R) \cap \text{BReach}_R(\orbitD_R)$. 
It is easy to see that the only non-trivial components of $R'$ are a $\orbitD$-degenerated and an $\orbitE$-degenerated component and thereby $R'$ is a ternary self-complementary $(\orbitD \rightarrow \orbitD)$-implication that contains all tuples required in Case~\ref{degen:essenternary}.

Observe also that the same reasoning works when either $R_1$ or $R_2$ has an essentially ternary $(\orbitE, \ldots, \orbitD)$-tuple. Thus, from now on we assume that neither of these tuples is present in any of these relations and that $\tau$
contains at least one essentially quaternary tuple. 

Consider now the case where there are at least two 
essentially quaternary tuples in $\tau$.  
By Observations~\ref{obs:selfcomplementary} and~\ref{obs:tuplescomposition}  we have that  $R \equiv (R_1 \circ R_2)^k$ is a self-complementary $(\subsetA \rightarrow \subsetA)$-implication  which contains a fully free  $(\orbitE, \ldots, \orbitD)$-tuple. 
Since $\orbitC \notin \subsetA$, by Observation~\ref{obs:minmaxcomponents}, we have that $\bipartite_{R_1, R_2}$  has at least three pairwise different $\orbitO_i$-degenerated components and two of them are in the $(\subsetA, \subsetA)$-subgraph. The same holds for $\bipartite_{R,R}$.
Without loss of generality assume that the $(\orbitO_1, \orbitO_1)$-component is maximal in the $(\subsetA, \subsetA)$-subgraph, $(\orbitO_2, \orbitO_2)$ is minimal in  the $(\subsetA, \subsetA)$-subgraph and
$(\orbitO_3, \orbitO_3)$ is minimal in the 
$((\Pi_{1,2}(R) \setminus \subsetA), (\Pi_{1,2}(R) \setminus \subsetA))$-subgraph of $\bipartite_{R,R}$.
Now, if $\orbitO_3$ is anti-reflexive, then we are in Case~\ref{degen:partfree}
already when we see $R$ as a self-complementary $(\subsetA \rightarrow \subsetA)$-implication. Otherwise, we are in Case~\ref{degen:partfree} with $R$ seen as a 
self-complementary $(\orbitO_1 \rightarrow \orbitO_1)$-implication.

The last case is where all but one non-degenerated tuple in $\tau$ are essentially ternary. First we assume that $\pi$ is of essential length at least $3$. Thus, there must be three consecutive tuples 
$t_{i}, t_{i+1}, t_{i+2}$ of $\tau$ such that  at least one of them is essentially quaternary and two remaining are essentially ternary. By Observation~\ref{obs:tuplescomposition} we have what follows. If there is $j \in \{0,1\}$ such that $t_{i+j}$ is essentially quaternary, $t_{i+j+1}$ is essentially ternary and $\orbitO_{i+j+2}$ is not $=$ or there is 
$j \in [2]$ such that $t_{i+j}$ is essentially quaternary, $t_{i+j-1}$ is essentially ternary and $\orbitO_{i+j-1}$ is not $=$, then $R_1 \circ R_2$ in case $i$ is odd or $R_2 \circ R_1$ in case $i$ is even contains a partially-free tuple. Otherwise one of two cases holds:
\begin{itemize}
\item $t_i$ is essentially quaternary, $t_{i+1}$ essentially ternary and $\orbitO_{i+2}$ is $=$, or
\item $t_{i+2}$ is essentially quaternary, $t_{i+1}$ essentially ternary and $\orbitO_{i+1}$ is $=$.
\end{itemize}
Again, by Observation~\ref{obs:tuplescomposition}, we have that either $R_1 \circ R_2$ or $R_2 \circ R_1$ has  an essentially quaternary tuple. Clearly neither $\orbitO_{i+3}$ in the first case nor $\orbitO_{i}$ in the second is $=$. Furthermore both $t_{i+2}$ in the first case and $t_i$ in the second case are essentially ternary. It follows that $R_1 \circ R_2 \circ R_1$ in a case $i$ is odd or $R_2 \circ R_1 \circ R_2$ if $i$ is even contain a partially free tuple. By Observations~\ref{obs:selfcomplementary}, \ref{obs:circbipartite} and~\ref{obs:tuplescomposition}, the relation $R \equiv (R_1 \circ R_2)^k$ is a self-complementary $(\orbitA \rightarrow \orbitA)$-implication that contains a partially-free tuple. It is also easy to see that all non-trivial degenerated components in $\bipartite_{R_1, R_2}$ are also in $\bipartite_{R, R}$. By Observation~\ref{obs:minmaxcomponents} we have that the former as well as the latter graph 
has at least three pairwise different $\orbitO_i$-degenerated components so that two of them are in $(\subsetA, \subsetA)$. Now, we complete the reasoning
exactly in the same way as in the previous case.

The only case left to be considered is when $\tau$ has exactly two non-degenerated tuples. In particular we can assume that a direct path from $\orbitE_L$ to $\orbitD_L$ that goes through $\orbitC_R$ is of essential length $2$. 
Thus we have a  $(\orbitE, \ldots, \orbitC)$-tuple in $R_1$ and a  
$(\orbitC, \ldots, \orbitD)$-tuple in $R_2$. One of these tuples  is  essentially ternary and the other is essentially quaternary. Without loss of generality assume that the former is essentially ternary.
Observe now that 
$$R'_1(x_1, x_2, x_3, x_4) \wedge x_2 = x_3 \wedge \subsetC(x_1, x_2) \wedge \subsetD(x_3,x_4)$$ where
$\subsetC = \text{FReach}_L(\orbitE_L) \cap \text{BReach}_L(\orbitD_L)$ and
$\subsetD = \text{FReach}_R(\orbitE_R) \cap 
\text{BReach}_R(\orbitD_R)$ contains a non-degenerated tuple.  In order to complete the proof of the proposition we will show that $R'_1$ is not connected, i.e.,the graph $\bipartite_{R,R'}$ where $R'(x_1, x_2, x_3, x_4) \equiv R(x_4, x_3, x_2, x_1)$ is not connected,
in particular that there is no path from $\orbitE_L$ to $\orbitD_R$. 
Observe that it is enough to show that there is no sequence of tuples 
$t_1, \ldots, t_{2k+1}$
such that 
$t_i$ is a $(\orbitO_i, \ldots, \orbitO_{i+1})$-tuple in $R_1$ when $i$ is odd and 
such that 
$t_i$ is a $(\orbitO_{i+1}, \ldots, \orbitO_{i})$-tuple in $R'_1$ when $i$ is even
and such that $\orbitO_1$ is $\orbitD$ and $\orbitO_{2k}$ is $\orbitE$.
Recall that $R_1$ has no $(\orbitE, \ldots, \orbitD)$-tuple and hence $k > 1$. It follows that there is a $(\orbitP, \ldots, \orbitP')$-tuple in $R_1$ with 
$\orbitP \subseteq \text{FReach}_L(\orbitE_L) \cap \text{BReach}_L(\orbitD_L)$ and
$\orbitP' \subseteq \text{FReach}_L(\orbitE_L) \cap \text{BReach}_L(\orbitD_L)$ and such that 
$\{ \orbitP, \orbitP' \} \cap \{ \orbitE, \orbitD \} = \emptyset$.
It implies that there is a direct path from $\orbitE_L$ do $\orbitD_L$ that goes through $\orbitP_L, \orbitP'_R$. It has to be of 
essential length at least $3$ which contradicts the assumption. We have that one of the three cases in the formulation of the propositions follows.
\end{proof}

We will now show that in every case from the formulation of Proposition~\ref{prop:degenreduce}, the structure $\structB$ is not preserved by a chain of quasi directed  Jonsson operations.

\noindent


\begin{lemma}
\label{lem:degen_nonconnected}
Let $\structB$ be a first-order reduct of a finitely bounded homogenous symmetric binary core whose age has free amalgamation. If $\structB$  pp-defines a quaternary relation $R$ which is not connected and which contains a non-degenerated $(\orbitA, \ldots, \orbitB)$-tuple, then 
$\structB$ is not preserved by any chain of  quasi directed Jonsson operations. 
\end{lemma}

\begin{proof}
Let $R'(x_1, x_2, x_3, x_4) \equiv R(x_4, x_3, x_2, x_1)$,
and let $(\subsetA, \subsetB)$-component be a component of $\bipartite_{R,R'}$
such that $\orbitA \subseteq \subsetA$ and $\orbitB \subseteq \subsetB$. Since $R$ is not connected, it is  a $(\subsetA \rightarrow \subsetB)$-implication, $R'$ is a $(\subsetB \rightarrow \subsetA)$-implication and $R,R'$ are complementary. Since $(\subsetA, \subsetB)$ is non-degenerated it follows by Theorem~\ref{thm:nondegen}  that $\structB$ is not preserved by any chain of quasi directed  J\'{o}nsson operations.
\end{proof}


\noindent
We now turn to Case~\ref{degen:essenternary} in the formulation of Proposition~\ref{prop:degenreduce}.

\begin{lemma}
\label{lem:ternarynoJonsson}
Let  $\structB$ be a first-order expansion of a finitely bounded homogeneous symmetric binary core whose age has free amalgamation. If $\structB$ pp-defines a ternary self-complementary  relation $R:\subsetA \rightarrow \subsetA$ which contains:
\begin{itemize}
\item a $\orbitA$-degenerated component,
 a $\orbitB$-degenerated component for some orbitals $\orbitA \subseteq \subsetA$, $\orbitB \nsubseteq \subsetA$, and
\item a $(\orbitB,  \orbitD, \orbitA)$-tuple for some anti-reflexive orbital $\orbitD$.
\end{itemize}
then $\structB$ is not preserved by any chains of quasi  directed J\'{o}nsson operations. 
\end{lemma}

\begin{proof}
The proof goes along the lines of similar proofs in the previous subsection. Except for a trivial case of $(3,3)$-constant pairs of triples, we use  Definition~\ref{def:UVconstantAB} only for $([3],[2])$-constant and $([3], \{ 2,3\})$-constant pairs $u,v$ of triples in $A$ satisfying $\orbitA\orbitA\orbitB$ and $\orbitA\orbitB\orbitB$ respectively with a difference to the original definition such that 
$(v[2], v[3]) \in \orbitD$ in the first case and $(v[1],v[2]) \in\orbitD$ in the second case. There is a $(\orbitB, \orbitD, \orbitA)$-tuple $t$ in $R$, and hence   both kinds of $(U,V)$-constant pairs of triples mentioned above exist in $A$.

We assume on the contrary  that
there is a chain $(D_1, \ldots,D_n)$ of quasi directed J\'{o}nsson operations preserving $\structB$ and $R$. We  show that it contradicts the fact that $\Pol(\structB)$ preserves $\subsetA$ which is $pp$-definable . 

\begin{claim}
\label{claim:ternarynoJonsson}
For all $i \in [n]$ we have both of the following.
\begin{itemize}
\item  For all $([3],[2])$-constant pairs of triples $(u,v)$ in $A$ 
satisfying $\orbitA \orbitA \orbitB(u,v)$ we have that $(D_i(u), D_i(v)) \in \subsetA$. 
\item For all $([3],\{2,3\})$-constant pairs of triples $(u,v)$ in $A$ 
satisfying $\orbitA\orbitB\orbitB(u,v)$ we have that  $(D_i(u), D_i (v))\in \subsetA$. 
\end{itemize}
\end{claim}

\begin{figure}
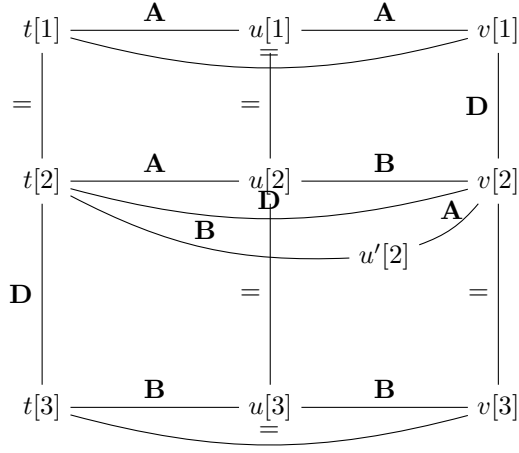

\tikz {
\node (t1) at (4,4) {$t[1]$};
\node (t2) at (4,2) {$t[2]$};
\node (t3) at (4,-1) {$t[3]$};
\node (u1) at (7,4) {$u[1]$};
\node (u2) at (7,2) {$u[2]$};
\node (u3) at (7,-1) {$u[3]$};
\node (v1) at (10,4) {$v[1]$};
\node (v2) at (10,2) {$v[2]$};
\node (v3) at (10,-1) {$v[3]$};
\node (u2prim) at (8.5 ,1) {$u'[2]$};

\draw (t1) -- node[above] {$\orbitA$} (u1);
\draw (t2) -- node[above] {$\orbitA$} (u2);
\draw (t3) -- node[above] {$\orbitB$} (u3);

\draw (u1) -- node[above] {$\orbitA$} (v1);
\draw (u2) -- node[above] {$\orbitB$} (v2);
\draw (u3) -- node[above] {$\orbitB$} (v3);

\draw (t1) -- node[left] {$=$} (t2);
\draw (t2) -- node[left] {$\orbitD$} (t3);

\draw (u1) -- node[left] {$=$} (u2);
\draw (u2) -- node[left] {$=$} (u3);

\draw (v1) -- node[left] {$\orbitD$} (v2);
\draw (v2) -- node[left] {$=$} (v3);

\path (u2prim) edge[bend right=15] node[above] {$\orbitA$} (v2);
\path (u2prim) edge[bend left=15] node[above] {$\orbitB$} (t2);

\path (t1) edge[bend right=15] node[above] {$=$} (v1);
\path (t2) edge[bend right=15] node[above] {$\orbitD$} (v2);
\path (t3) edge[bend right=15] node[above] {$=$} (v3)
}
\centering
\caption{The figure needed for a proof of the second item in Claim~\ref{claim:ternarynoJonsson}. }
\label{fig:ternaryextension}
\end{figure}

\smallskip
\noindent
\begin{proof}
We prove the claim by the induction on $i \in [n]$. 

\textbf{(BASE CASE)} For the first bullet 
we start with a  $([3],[2])$-constant pair of triples $(u,v)$ in $A$ 
satisfying $\orbitA \orbitA \orbitB(u,v)$. If we replace $v$ with $v'$ which is identical to $v$ with a difference that  $v'[3] = v[2]$.
It is clear that $(u,v')$ is a $([3],[3])$-constant pair of triples in $A$ satisfying $\orbitA\orbitA\orbitA$, and hence $(D_1(u), D_1(v')) \in \subsetA$. By~(\ref{eq:D1}) we have $D_1(v') = D_1(v) $. It follows  $(D_1(u), D_1(v)) \in \subsetA$.

For the second bullet in the formulation of the claim,
we start with a $([3],\{2,3\})$-constant pair of triples $(u,v)$ 
satisfying $\orbitA \orbitB \orbitB (u,v)$ and extend it with $t \in A^3$ as pictured in Figure~\ref{fig:ternaryextension}.
Observe that no new elements are needed, and hence the extension is trivial. Indeed, we have that $t[1] = t[2] = v[2]$ and $t[3] = v[3]$. There also exists $u'$ identical with $u$ with a difference that $u'[2] \neq u[2]$  so that 

\begin{itemize}
\item $a_1 =(t[1], u[1], v[1])$ is a $(\orbitA, =, \orbitA)$-tuple in $R$, 
\item $a_2 = (t[2], u'[2], v[2])$ is a $(\orbitB, \orbitD, \orbitA)$-tuple in $R$, and
\item $a_3 = (t[3], u[3], v[3])$ is a $(\orbitB, =, \orbitB)$-tuple in $R$.
\end{itemize}
Indeed, since the age of $\structA$ has free amalgamation we have that 
there are  $a,b,c,d \in A$ such that $(a,b,c)$ is a $(\orbitA, \orbitD, \orbitB)$-tuple, $(a,d,c)$ is a $(\orbitB, \orbitD, \orbitA)$-tuple and $(b,d) \in \orbitN$. 
If $\alpha$ is an automorphism of $\structA$ sending   
$(t[2], u[2], v[2])$ to $(a,b,c)$, then we may take $\alpha^{-1}(d)$ for $u'[2]$.

By the first bullet we have that $(D_1(t),D_1(u)) \in \subsetA$, by~(\ref{eq:Di}) that $(D_1(t), D_1(u')) \in \subsetA$. Since $R$ is a $(\subsetA \rightarrow \subsetA)$-implication, it follows that $(D_1(u',v)) \in \subsetA$ and $(D_1(u,v)) \in \subsetA$. It completes the proof of the base case.

\smallskip
\noindent
\textbf{(INDUCTION STEP)}
For the induction step, consider any $([3],[2])$-constant pair of triples $(u,v)$
in $A$
satisfying $\orbitA \orbitA \orbitB (u,v)$. We want to show that
$(D_{i+1}(u), D_{i+1}(v)) \in \subsetA$.
To that end consider  $v'$ such that $v'[1] =  v[1]$ and $v'[2] = v'[3] = v[3]$.
Observe that $(u,v')$ is a $([3],\{ 2,3 \})$-constant pair of triples in $A$ satisfying $\orbitA \orbitB \orbitB (u,v')$. By the induction hypothesis, it follows that ($D_i(u),D_i(v')) \in \subsetA$.
By~(\ref{eq:Dii+1}), we have $D_{i}(v) = D_{i+1}(v')$
and clearly $D_{i+1}(u) = D_i(u)$. It follows that 
$(D_{i+1}(u), D_{i+1}(v)) \in \subsetA$.
Now, as in the base case we show that
$(D_{i+1}(u), D_{i+1}(v)) \in \orbitA$ for all $([3], \{ 2,3 \})$-constant 
pair of triples $(u,v)$ in $\orbitA$ satisfying $\orbitA \orbitB \orbitB(u,v)$. It completes the proof of the claim.  
\end{proof}

\begin{figure}
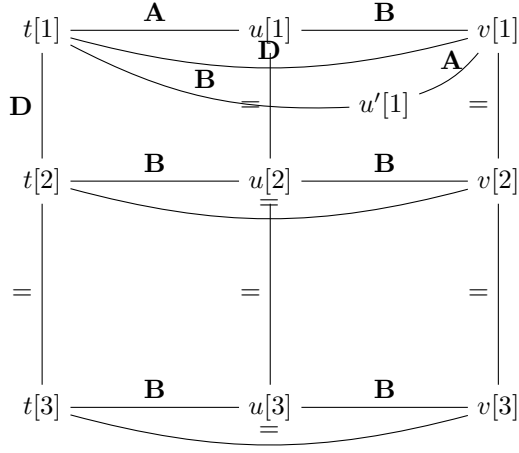

\tikz {
\node (t1) at (4,4) {$t[1]$};
\node (t2) at (4,2) {$t[2]$};
\node (t3) at (4,-1) {$t[3]$};
\node (u1) at (7,4) {$u[1]$};
\node (u2) at (7,2) {$u[2]$};
\node (u3) at (7,-1) {$u[3]$};
\node (v1) at (10,4) {$v[1]$};
\node (v2) at (10,2) {$v[2]$};
\node (v3) at (10,-1) {$v[3]$};
\node (u1prim) at (8.5 ,3) {$u'[1]$};

\draw (t1) -- node[above] {$\orbitA$} (u1);
\draw (t2) -- node[above] {$\orbitB$} (u2);
\draw (t3) -- node[above] {$\orbitB$} (u3);

\draw (u1) -- node[above] {$\orbitB$} (v1);
\draw (u2) -- node[above] {$\orbitB$} (v2);
\draw (u3) -- node[above] {$\orbitB$} (v3);

\draw (t1) -- node[left] {$\orbitD$} (t2);
\draw (t2) -- node[left] {$=$} (t3);

\draw (u1) -- node[left] {$=$} (u2);
\draw (u2) -- node[left] {$=$} (u3);

\draw (v1) -- node[left] {$=$} (v2);
\draw (v2) -- node[left] {$=$} (v3);

\path (u1prim) edge[bend right=15] node[above] {$\orbitA$} (v1);
\path (u1prim) edge[bend left=15] node[above] {$\orbitB$} (t1);

\path (t1) edge[bend right=15] node[above] {$\orbitD$} (v1);
\path (t2) edge[bend right=15] node[above] {$=$} (v2);
\path (t3) edge[bend right=15] node[above] {$=$} (v3)
}
\centering
\caption{A figure which we use in the final step of the proof of Lemma~\ref{lem:ternarynoJonsson}.}
\label{fig:ternaryfinalstep}
\end{figure}

For the final step of the proof of the lemma consider  a $([3], [3])$-constant pair of triples $u,v$ satisfying $\orbitB \orbitB \orbitB (u,v)$
and  $t \in A^3$ so that $u,t$ forms a $(\{2,3\},[3])$-constant pair of triples satisfying $\orbitA\orbitB\orbitB(u,v)$ as in Figure~\ref{fig:ternaryfinalstep}. By the assumption we have that there is a $(\orbitB, \orbitD, \orbitA)$-tuple $(c,a,b)$ in $A$. Set $\alpha$ to be an automorphism of $\structA$ sending $(a,b)$ to $(u[1], v[1])$. Observe that we may take $(\alpha(c),\alpha(b),\alpha(b))$ for $t$.

Let now $u'$ be as $u$ with a difference that $u'[1]$ satisfying 
\begin{itemize}
\item $a_1 =(t[1], u'[1], v[1])$ is a $(\orbitB, \orbitD,  \orbitA)$-tuple, 
\item $a_2 = (t[2], u[2], v[2])$ is a $(\orbitB, =, \orbitB)$-tuple, and
\item $a_3 = (t[3], u[3], v[3])$ is a $(\orbitB, =, \orbitB)$-tuple.
\end{itemize}
as well as $(u'[1], u[i]) \in \orbitN$ for $i \in \{ 2, 3 \}$.
Again, we use the fact that the age of $\structA$ has free amalgamation in order to  see that there are 
$a,b,c,d \in A$ such that $(a,b,c)$ is a $(\orbitA, \orbitD, \orbitB)$-tuple, $(a,d,c)$ is a $(\orbitB, \orbitD, \orbitA)$-tuple and $(b,d) \in \orbitN$. Sending $(t[1], u[1], v[1])$ by an automorphism $\alpha$ to $(a,b,c)$ we get that $u'[1]$ is $\alpha^{-1}(d)$. Thus, it exists.

By~(\ref{eq:Dn}) we have that $(D_n(t), D_n(u')) \in \subsetA$. 
Since all  $a_i$ with $i \in [3]$ are in $R$ and $R$ is a $(\subsetA \rightarrow \subsetA)$-implication, it follows that  $(D_n(u'), D_n(v)) \in \subsetA$. Again, by~(\ref{eq:Dn}), we have $(D_n(u), D_n(v)) \in \subsetA$.
It contradicts the assumption that $D_n$
preserves $\orbitB$ and completes the proof of the lemma. 
\end{proof}

\noindent
We finally turn to the third case in Proposition~\ref{prop:degenreduce}.

\begin{lemma}
\label{lem:partfreenoJonsson}
Let $\structB$ be a  first-order expansion of a finitely bounded homogeneous symmetric binary core $\structA$ whose age has free amalgamation.
If $\structB$ pp-defines a self-complementary implication $R :\subsetA \rightarrow \subsetA$ 
such that $\bipartite_{R, R}$ contains a $\orbitA$-degenerated
component and a $\orbitB$-degenerated component where $\orbitA \subseteq \subsetA$ and $\orbitB \nsubseteq \subsetA$ are both anti-reflexive and a 
partially-free 
$(\orbitO_{11}, \orbitO_{12}, \orbitO_{13},\orbitN, \orbitO_{23}, \orbitO_{24}, \orbitO_{34} )$-tuple is in $R$,  
then $\structB$ is not preserved by any chain of quasi directed J\'{o}nsson operations.  
\end{lemma}

\begin{figure}
\tikz {
\node (s1) at (0,5) {$s[1]$};
\node (s2) at (0,3) {$s[2]$};
\node (s3) at (0,0) {$s[3]$};
\node (t1) at (2,5) {$t[1]$};
\node (t2) at (2,3) {$t[2]$};
\node (t3) at (2,0) {$t[3]$};
\node (u1) at (4,5) {$u[1]$};
\node (u2) at (4,3) {$u[2]$};
\node (u3) at (4,0) {$u[3]$};
\node (v1) at (6,5) {$v[1]$};
\node (v2) at (6,3) {$v[2]$};
\node (v3) at (6,0) {$v[3]$};
\node (t2prim) at (1,1.2) {$t'[2]$};
\node (u2prim) at (5,1.2) {$u'[2]$};

\draw (s1) -- node[above] {$\orbitA$} (t1);
\draw (s2) -- node[above] {$\orbitA$} (t2);
\draw (s3) -- node[above] {$\orbitB$} (t3);

\draw (t1) -- node[above] {$=$} (u1);
\draw (t2) -- node[above] {$=$} (u2);
\draw (t3) -- node[above] {$=$} (u3);

\draw (u1) -- node[above] {$\orbitA$} (v1);
\draw (u2) -- node[above] {$\orbitB$} (v2);
\draw (u3) -- node[above] {$\orbitB$} (v3);

\draw (s1) -- node[left] {$=$} (s2);
\draw (s2) -- node[left] {$\orbitN$} (s3);

\draw (t1) -- node[left] {$=$} (t2);
\draw (t2) -- node[left] {$=$} (t3);

\draw (u1) -- node[left] {$=$} (u2);
\draw (u2) -- node[left] {$=$} (u3);

\draw (v1) -- node[left] {$\orbitN$} (v2);
\draw (v2) -- node[left] {$=$} (v3);

\path (u2prim) edge[bend right=20] node[above] {$\orbitO_{34}$} (v2);
\path (t2prim) edge[bend left=20] node[above] {$\orbitO_{11}$} (s2);
\path (u2prim) edge[bend left=20] node[above] {$\orbitO_{23}$} (t2prim);

\path (s1) edge[bend right=15] node[above] {$=$} (v1);
\path (s2) edge[bend right=15] node[above] {$\orbitN$} (v2);
\path (s3) edge[bend right=15] node[above] {$=$} (v3)
}
\centering
\caption{A substructure of $\structA$ used in the proof of Claim~\ref{claim:partfreenoJonsson}.}
\label{fig:partfreeextension}
\end{figure}

\begin{proof}
Again, we assume on the contrary that there is a chain $(D_1, \ldots,D_n)$ of quasi directed J\'{o}nsson operations preserving $\structB$ and $R$ and  as usual we show that it contradicts the fact that $\Pol(\structB)$ preserves all orbitals. 

In contrast to the proof of Lemma~\ref{lem:ternarynoJonsson}, we use the original version of Definition~\ref{def:UVconstantAB}.

We provide a claim as usual.

\begin{claim}
\label{claim:partfreenoJonsson}
For all $i \in [n]$ we have both of the following.
\begin{itemize}
\item  For all $([3],[2])$-constant pairs of triples $(u,v)$ in $A$ 
satisfying $\orbitA \orbitA \orbitB(u,v)$ we have that $(D_i(u), D_i(v)) \in \subsetA$. 
\item For all $([3],\{2,3\})$-constant pairs of triples $(u,v)$ in $A$ 
satisfying $\orbitA\orbitB\orbitB(u,v)$ we have that  $(D_i(u), D_i (v)) \in \subsetA$. 
\end{itemize}
\end{claim}

\smallskip
\noindent
\begin{proof}
We prove the claim by the induction on $i \in [n]$. 

\textbf{(BASE CASE)} 
The first bullet is a special case of what we proved for  Claim~\ref{claim:bnbnoJonsson}.
For the second bullet consider a $([3],\{2,3\})$-constant pair of triples $u,v$ in $A$
satisfying $\orbitA \orbitB \orbitB (u,v)$ and a $([3], [2])$-constant pair of triples $s,t$ satisfying $\orbitA \orbitA \orbitB (t,s)$ as in Figure~\ref{fig:partfreeextension}. Observe that all values in 
vectors $s, t$ are already in $u,v$ and hence $s,t$ exist.

We will now show that there exist $t', u'$ identical to  $t,u$ but with a difference that  $t'[2], u'[2]$ are new and such that
\begin{itemize}
\item $a_1 =(s[1], t[1], u[1], v[1])$ is a $(\orbitA, \orbitA, =, =, \orbitA, \orbitA)$-tuple, 
\item $a_2 = (s[2], t'[2], u'[2], v[2])$ is a $(\orbitO_{12}, \orbitO_{13},    \orbitN, \orbitO_{23}, \orbitO_{24}, \orbitO_{34})$-tuple, and
\item $a_3 = (s[3], t[3], u[3], v[3])$ is a $(\orbitB, \orbitB, =,= \orbitB, \orbitB)$-tuple.
\end{itemize}
Indeed, since the age of $\structA$ has free amalgamation there are elements $a,b,c,d,e \in A$ such that $(a,b,c)$ is a $(\orbitA, \orbitN,\orbitB)$-tuple, $(a,d,e,c)$  is a $(\orbitO_{12}, \orbitO_{13},    \orbitN, \orbitO_{23}, \orbitO_{24}, \orbitO_{34})$-tuple and all other orbitals between the elements $a,b,c,d,e$ either follow from semi-transitivity of equality or are $\orbitN$. If $\alpha$ is an  automorphism of $\structA$ sending $(a,b,c)$ to $(s[2], t[2],v[2])$, then we take $\alpha^{-1}(d)$ for $t'[2]$ and $\alpha^{-1}(e)$ for $u'[2]$.

By the first item we have that
$(D_1(s),D_1(t)) \in \subsetA$.
By~(\ref{eq:Di}), it follows that $(D_1(s),D_1(t')) \in \subsetA$
Since $a_i \in R$ for $i \in [3]$
we have $(D_1(u'),D_1(v)) \in \subsetA$.
By~(\ref{eq:Di}), we have that $(D_1(u),D_1(v)) \in \subsetA$, which was to be proved. It completes the proof for the base case.


\smallskip
\noindent
\textbf{(INDUCTION STEP)}
The first bullet we prove exactly like in the proof of 
Claim~\ref{claim:bnbnoJonsson}, and the second bullet as in the base case.
\end{proof} 

In the final step, as usual, we show that under the working assumptions  
there is a $(3,3)$-constant part of triples $(u,v)$ satisfying $\orbitB \orbitB \orbitB$ such that $(D_n(u), D_n(v)) \in \subsetA$. 
We provide the reasoning analogous to reasoning in the proof of the claim with a difference that we use Figure~\ref{fig:partfreefinalstep} instead of Figure~\ref{fig:partfreeextension}. 
\end{proof}

\begin{figure}
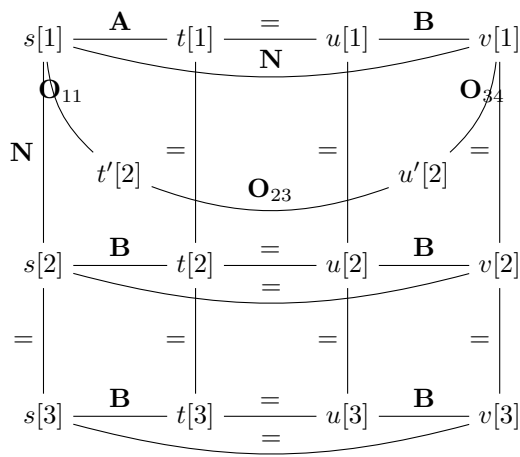

\tikz {
\node (s1) at (0,5) {$s[1]$};
\node (s2) at (0,2) {$s[2]$};
\node (s3) at (0,0) {$s[3]$};
\node (t1) at (2,5) {$t[1]$};
\node (t2) at (2,2) {$t[2]$};
\node (t3) at (2,0) {$t[3]$};
\node (u1) at (4,5) {$u[1]$};
\node (u2) at (4,2) {$u[2]$};
\node (u3) at (4,0) {$u[3]$};
\node (v1) at (6,5) {$v[1]$};
\node (v2) at (6,2) {$v[2]$};
\node (v3) at (6,0) {$v[3]$};
\node (t1prim) at (1,3.2) {$t'[2]$};
\node (u1prim) at (5,3.2) {$u'[2]$};

\draw (s1) -- node[above] {$\orbitA$} (t1);
\draw (s2) -- node[above] {$\orbitB$} (t2);
\draw (s3) -- node[above] {$\orbitB$} (t3);

\draw (t1) -- node[above] {$=$} (u1);
\draw (t2) -- node[above] {$=$} (u2);
\draw (t3) -- node[above] {$=$} (u3);

\draw (u1) -- node[above] {$\orbitB$} (v1);
\draw (u2) -- node[above] {$\orbitB$} (v2);
\draw (u3) -- node[above] {$\orbitB$} (v3);

\draw (s1) -- node[left] {$\orbitN$} (s2);
\draw (s2) -- node[left] {$=$} (s3);

\draw (t1) -- node[left] {$=$} (t2);
\draw (t2) -- node[left] {$=$} (t3);

\draw (u1) -- node[left] {$=$} (u2);
\draw (u2) -- node[left] {$=$} (u3);

\draw (v1) -- node[left] {$=$} (v2);
\draw (v2) -- node[left] {$=$} (v3);

\path (u1prim) edge[bend right=20] node[above] {$\orbitO_{34}$} (v1);
\path (t1prim) edge[bend left=20] node[above] {$\orbitO_{11}$} (s1);
\path (u1prim) edge[bend left=20] node[above] {$\orbitO_{23}$} (t1prim);

\path (s1) edge[bend right=15] node[above] {$\orbitN$} (v1);
\path (s2) edge[bend right=15] node[above] {$=$} (v2);
\path (s3) edge[bend right=15] node[above] {$=$} (v3)
}
\centering
\caption{A substructure of $\structA$ used in the final step of the proof of Lemma~\ref{lem:partfreenoJonsson}}
\label{fig:partfreefinalstep}
\end{figure}

\noindent
We will now conclude Section~\ref{sect:alldegen} by a theorem which follows directly from Proposition~\ref{prop:degenreduce} and Lemmas~\ref{lem:degen_nonconnected},~\ref{lem:ternarynoJonsson} and~\ref{lem:partfreenoJonsson}.

\begin{theorem}
\label{thm:degen}
Let $\structB$ be a first-order expansion of a finitely bounded homogeneous symmetric binary core which pp-defines complementary implications $R_1:\subsetA \rightarrow \subsetB$ and $R_2: \subsetB \rightarrow \subsetA$ with $\subsetA \neq \subsetB$ and such that every non-trivial component of $\bipartite_{R_1, R_2}$ is degenerated. Then $\structB$ is not preserved by any chains of quasi directed J\'{o}nnson operations.  
\end{theorem}

\section{Main Result}
\label{sect:mainresult}

We are ready to prove the main theorem of the paper.

\smallskip
\noindent
\textbf{Proof of Theorem~\ref{thm:main}}
The relational clone of $\structB$ is clearly either implicationally uniform or not. In the former case the theorem holds by Theorem~\ref{thm:impuniform} while in the latter $\structB$ pp-defines complementary implications $R_1:\subsetA \rightarrow \subsetB, R_2: \subsetB \rightarrow \subsetA$ such that $\subsetA \neq \subsetB$. Now, either 
$\bipartite_{R_1, R_2}$ contains a non-trivial non-degenerated component or all non-trivial components of this graph are degenerated. In either case, the structure $\structB$ is not preserved by any chain of quasi directed J\'{o}nsson operations by Theorems~\ref{thm:nondegen} or~\ref{thm:degen}.
$\qed$

\bibliographystyle{acm}
\bibliography{mybib}

\appendix

\end{document}